\documentclass[11pt,draftcls,onecolumn] {IEEEtran}
\usepackage{amsmath,graphicx}
\usepackage{cite}
\usepackage{amsfonts}
\usepackage{fancyhdr}
\usepackage{amsmath,amssymb,amsthm}%
\usepackage{graphicx}%
\usepackage{listings}%
\usepackage{arydshln}
\usepackage{subfig}%
\usepackage{algorithm,algorithmic}
\newtheorem{mytheo}{Theorem}
\newtheorem{myrem}{Remark}%
\newtheorem{mylem}{Lemma}%

\renewcommand{\thefootnote}{\fnsymbol{footnote}}
\DeclareMathOperator*{\argmin}{arg\,min}

\def\x{{\mathbf x}}

\begin{document}

\title{Rigid Body Localization Using Sensor Networks: Position and Orientation Estimation}
                \author{Sundeep Prabhakar Chepuri,~\IEEEmembership{Student Member,~IEEE,} \\Geert~Leus,~\IEEEmembership{Fellow,~IEEE,} and Alle-Jan van der Veen,~\IEEEmembership{Fellow,~IEEE}
                
\thanks{This work was supported in part by STW under the FASTCOM project (10551) and in part by NWO-STW under the VICI program (10382).}

\thanks{All the authors are with the Faculty of Electrical Engineering, Mathematics and Computer Science, Delft University of Technology, The Netherlands. 
Email:~\{s.p.chepuri;g.j.t.leus;a.j.vanderveen\}@tudelft.nl.}
\thanks{A conference precursor of this manuscript has been
published in~\cite{ICASSP13chepuri}.}
}
\markboth{IEEE TRANSACTIONS ON SIGNAL PROCESSING (Draft)}
{Chepuri \MakeLowercase{\text \it{et al.}}: Rigid body localization using sensor networks}
\maketitle
\thispagestyle{empty}
\pagenumbering{arabic}
\renewcommand{\thefootnote}{\arabic{footnote}}
\IEEEoverridecommandlockouts
\maketitle
\begin{abstract}
In this paper, we propose a novel framework called rigid body localization for joint position and orientation estimation of a rigid body. We consider a setup in which a few sensors are mounted on a rigid body. The absolute position of the sensors on the rigid body, or the absolute position of the rigid body itself is not known. However, we know how the sensors are mounted on the rigid body, i.e., the sensor topology is known.  Using range-only measurements between the sensors and a few anchors (nodes with known absolute positions), and without using any inertial measurements (e.g., accelerometers), we estimate the position and orientation of the rigid body. For this purpose, the absolute position of the sensors is expressed as an affine function of the Stiefel manifold. In other words, we represent the orientation as a rotation matrix, and absolute position as a translation vector. We propose a least-squares (LS), simplified unitarily constrained LS (SUC-LS), and optimal unitarily constrained least-squares (OUC-LS) estimator, where the latter is based on Newton's method. As a benchmark, we derive a unitarily constrained Cram\'er-Rao bound (UC-CRB). The known topology of the sensors can sometimes be perturbed during fabrication. To take these perturbations into account, a simplified unitarily constrained total-least-squares (SUC-TLS), and an optimal unitarily constrained total-least-squares (OUC-TLS) estimator are also proposed.
\\\\
{\bf EDICS}: SEN-LOCL Source localization in sensor networks, SEN-APPL Applications of sensor networks,  SAM-APPL Applications of sensor and array multichannel processing, SPC-INTF Applications of sensor networks. 
\end{abstract}
\section{Introduction}
\IEEEPARstart{O}{}ver
the past decade, advances in wireless sensor technology have enabled the usage of wireless sensor networks (WSNs) in different areas related to sensing, monitoring, and control~\cite{WSNs}. Wireless sensors are nodes equipped with a radio transceiver and a processor, capable of wireless communications and computational operations. A majority of the applications that use WSNs rely on a fundamental aspect of either associating the location information to the data that is acquired by spatially distributed sensors (e.g., in field estimation), or to identify the location of the sensor itself (e.g., in  security, rescue, logistics). Identifying the sensor's location is a well-studied topic~\cite{localizationSPM,GinnakisLoc,Gusta05SPM}, and it is commonly referred to as \textit{localization}. 

Localization can be either absolute or relative. In absolute localization, the aim is to estimate the absolute position of the sensor(s) using a few reference nodes whose absolute positions are known, commonly referred to as {\it anchors}. Absolute localization problems are typically solved using measurements from certain physical phenomena, e.g., time-of-arrival (TOA), time-difference-of-arrival (TDOA), received signal strength (RSS), or angle-of-arrival (AOA)~\cite{localizationSPM,GinnakisLoc}.  Localization can also be relative, in which case the aim is to estimate the constellation of the sensors or the topology of the WSN, and determining the location of a sensor relative to the other sensors is sufficient. Classical solutions to relative localization are based on multi-dimensional scaling (MDS) using range measurements~\cite{MDS,MDS2,CostasMDS}. There exists a plethora of algorithms based on these two localization paradigms, and they recently gained a lot of interest to facilitate low-power and efficient localization solutions by avoiding global positioning system (GPS) based results and its familiar pitfalls.  

In this paper, we take a step forward from the classical localization, and provide a new and different flavor of localization, called \textit{rigid body localization}. In rigid body localization, we use a few sensors on a rigid body, and exploit the knowledge of the sensor topology to jointly estimate the position as well as the orientation of the rigid body. 
\subsection{Applications}
Rigid body localization has potential applications in a variety of fields. To list a few, it is useful in the areas of underwater (or in-liquid) systems, orbiting satellites, mechatronic systems, unmanned aircrafts, unmanned underwater vehicles, atmospheric flight vehicles, robotic systems, or ground vehicles.  In such applications, classical localization of the node(s) is not sufficient. For example, in an autonomous underwater vehicle (AUV)~\cite{underwaterAV}, or an orbiting satellite~\cite{olfar}, the sensing platform is not only subject to motion but also to rotation. Hence, next to position, determining the orientation of the body also forms a key component, and is essential for controlling, maneuvering, and monitoring purposes.  

The orientation is sometimes referred to as \textit{attitude} (aerospace applications) or \textit{tilt} (for industrial equipments and consumer devices). Traditionally, position and orientation are treated separately even though they are closely related. The orientation of a body is usually measured using inertial measurement units (IMUs) comprising of accelerometers~\cite{accelerometer}, gyroscopes and sometimes used in combination with GPS~\cite{GPSattitude}. However, IMUs generally suffer from accumulated errors often referred to as drift errors. Apart from IMUs, sensors like sun-trackers are sometimes used in satellites to measure orientation. 

On the other hand, in the presented rigid body localization approach we propose to use the communication packets containing the ranging information, just as in traditional localization schemes~\cite{localizationSPM,GinnakisLoc,Gusta05SPM}, to estimate both the position and the orientation of the rigid body. In short, we present rigid body localization as an estimation problem from a signal processing perspective.

\subsection{Contributions}

We propose a novel framework for joint position and orientation estimation of a rigid body in a three-dimensional space by borrowing techniques from classical absolute localization, i.e., using {\it range-only measurements} between all the sensor-anchor pairs. We consider a rigid body on which a few sensor nodes are mounted. These sensor nodes can be visualized as a sensor array. The absolute position of the sensors on the rigid body, or the absolute position of the rigid body itself is not known. However, the topology of how the sensors are mounted on the rigid body or the array geometry is known up to a certain accuracy. Based on the noisy {\it range-only} measurements between all the sensor-anchor pairs, we propose novel estimators for rigid body localization. More specifically, we propose a framework of rigid body localization as an {\it add-on} to the existing IMU based systems to correct for the drift errors or in situations where inertial measurements are not possible.

For this purpose, we express the orientation of the rigid body as a {\it rotation matrix} and the absolute position of the rigid body (instead of the absolute positions of all the sensors) as a {\it translation} vector, i.e., we represent the absolute position of the sensors as an affine function of the Stiefel manifold. We propose a {\it least-squares} (LS) estimator to jointly estimate the translation vector and the rotation matrix. Since rotation matrices are unitary matrices, we also propose a {\it simplified unitarily constrained least-squares} (SUC-LS) and {\it optimal unitarily constrained least-squares} (OUC-LS) estimator, both of which solve an optimization problem on the Stiefel manifold. We also derive a new {\it unitarily constrained Cram\'er-Rao bound} (UC-CRB), which is used as a benchmark for the proposed estimators. 

In many applications, the sensor topology might not be accurately known, i.e., the known topology of the sensor array can be noisy. Such perturbations are typically introduced while mounting the sensors during fabrication or if the body is not entirely rigid. To take such perturbations into account, we propose a {\it simplified unitarily constrained total-least-squares} (SUC-TLS) and an {\it optimal unitarily constrained total-least-squares} (OUC-TLS) estimator. The performance of the proposed estimators is analyzed using simulations. Using a sensor array with a known geometry not only enables orientation estimation, but also yields a better localization performance.

The framework proposed in this work is based on a static position and orientation, unlike most of the orientation estimators which are based on inertial measurements and a certain dynamical state-space model (e.g.,~\cite{dynamicalpose1}). Hence, our approach is useful when there is no dynamic model available. We should stress, however, that the proposed framework is believed to be suitable also for the estimation of dynamical position and orientation (tracking) using either a state-constrained Kalman filter or a moving horizon estimator (MHE), yet this extension is postponed to future work. 

\subsection{Outline and notations}
The remainder of this paper is organized as follows. The considered problem is described in Section~\ref{sec:probform}. In Section~\ref{sec:prel}, we provide preliminary information on classical LS based localization, and the Stiefel manifold, which are required to describe the newly developed estimators. The estimators based on perfect knowledge of the sensor topology and with perturbations on the known sensor topology are discussed in Section~\ref{sec:LS} and Section~\ref{sec:TLS}, respectively. In Section~\ref{sec:CRB}, we derive the unitarily constrained Cram\'er-Rao bound. Numerical results based on simulations are provided in Section~\ref{sec:sim}. The paper concludes with some remarks in Section~\ref{sec:con}.

The notations used in this paper are described as follows. Upper (lower) bold face letters are used for matrices (column vectors). $(\cdot)^T$ denotes transposition. $\mathrm{diag}(.)$ refers to a block diagonal matrix with the elements in its argument on the main diagonal. $\mathbf{1}_N$ ($\mathbf{0}_N$) denotes the $N \times 1$ vector of ones (zeros). $\mathbf{I}_N$ is an identity matrix of size $N$. $\mathbb{E}\{.\}$ denotes the expectation operation. $\otimes$ is the Kronecker product.  $(.)^\dag$ denotes the pseudo inverse, i.e., for a full column-rank tall matrix ${\bf A}$ the pseudo inverse (or the left-inverse) is given by ${\bf A}^\dag = ({\bf A}^T{\bf A})^{-1}{\bf A}^T$, and for a full row-rank wide matrix ${\bf A}$ the pseudo inverse (or the right-inverse) is given by ${\bf A}^\dag = {\bf A}^T({\bf A}{\bf A}^T)^{-1}$. The right- or left-inverse will be clear from the context. $\mathrm{vec}(.)$ is an $MN \times 1$ vector formed by stacking the columns of its matrix argument of size $M \times N$. $\mathrm{vec}^{-1}(.)$ is an $M \times N$ matrix formed by the inverse $\mathrm{vec}(.)$ operation on an $MN \times 1$ vector. Finally, $\mathrm{tr}(.)$ denotes the matrix trace operator.

\section{Problem formulation}\label{sec:probform}
\begin{figure} [!t]
\centering
\includegraphics[width=3in]{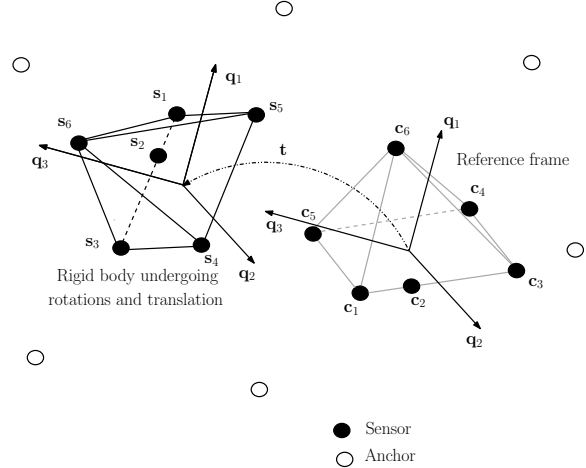}
\caption{An illustration of the sensors on a rigid body undergoing a rotation and translation.}
\label{fig:model}
\end{figure}
Consider a network with $M$ anchors (nodes with known absolute locations) and $N$ sensors in a $3$-dimensional space. The sensors are mounted on a rigid body as illustrated in Fig.~\ref{fig:model}. The absolute position of the sensors or the rigid body itself in the $3$-dimensional space is not known. The wireless sensors are mounted on the rigid body (e.g., at the factory), and the topology of how these sensors are mounted is known up to a certain accuracy. In other words, we connect a so-called {\it reference frame} to the rigid body, as illustrated in Fig. 1, and in that reference frame, the coordinates of the $n$th sensor are given by the known $3 \times 1$ vector ${\bf c}_n= [ c_{n,1}, c_{n,2}, c_{n,3} ]^T$. So the sensor topology is basically determined by the matrix ${\bf C} = [ {\bf c}_1, {\bf c}_2, \dots, {\bf c}_N ] \in {\mathbb R}^{3 \times N}$. 

Let the absolute coordinates of the $m$th anchor and the $n$th sensor be denoted by a $3 \times 1$ vector ${\bf a}_m$ and ${\bf s}_n$, respectively. These absolute positions of the anchors and the sensors are collected in the matrices ${\bf A} = [{\bf a}_1,{\bf a}_2,\ldots,{\bf a}_M] \in \mathbb{R}^{3 \times M}$ and ${\bf S} = [{\bf s}_1,{\bf s}_2,\ldots,{\bf s}_N] \in \mathbb{R}^{3 \times N}$, respectively. 

The pairwise distance (or the Euclidean distance) between the $m$th anchor and the $n$th sensor $r_{m,n}={\|{\bf a}_m-{\bf s}_n\|}_2$ is typically obtained by ranging~\cite{localizationSPM,ChepuriSPL}. The noisy range measurements can be expressed as
\begin{equation*}
\begin{aligned}
\label{eq:range}
{y}_{m,n}=r_{m,n} + v_{m,n}\\
 \end{aligned}
\end{equation*} where $v_{m,n}$ is the additive stochastic noise resulting from the ranging process. Assuming TOA-based ranging, we model $v_{m,n}$ as an independent and identically distributed (i.i.d.) zero mean white random process with variance
\begin{equation*}
\begin{aligned}
\label{eq:variance}
\sigma^2_{m,n} =  \frac{r_{m,n}^2}{(\sigma_r^2/\sigma_v^2)}= \frac{r_{m,n}^2}{\zeta},
\end{aligned}
\end{equation*}
where we define the reference range $\zeta = \frac{\sigma_v^2}{\sigma_r^2}$ with $\sigma_v^2$ being the reference ranging noise variance and $\sigma_r^2$ indicating the confidence on the range measurements. The reference ranging noise is the ranging noise when any two nodes are a unit distance apart, and this is nominally the same for all the anchors. The $r_{m,n}^2$ term penalizes the range measurements based on distance, and is due to the path-loss model assumption. 

The squared-range between the $m$th anchor and the $n$th sensor can be written as
\begin{equation*}
\begin{aligned}
\label{eq:rangesquare}
r_{m,n}^2&={\|{\bf a}_m - {\bf s}_n\|}_2^2 \\
&= {\|{\bf a}_m\|^2-2{\bf a}_m^T{\bf s}_n+\|{\bf s}_n\|^2}
\end{aligned}
\end{equation*} and the squared-range measurements as
\begin{equation}
\begin{aligned}
\label{eq:rangesquare}
 {d}_{m,n} = y_{m,n}^2&=r^2_{m,n} + 2r_{m,n}v_{m,n}+v_{m,n}^2\\
 &=  {r}_{m,n}^2 +n_{m,n},
\end{aligned}
\end{equation}
where $n_{m,n} = 2r_{m,n}v_{m,n}+v_{m,n}^2$ is the new noise term introduced due to the squaring of the range measurements. 

Under the condition of sufficiently small errors and ignoring the higher-order terms, we can approximate the stochastic properties of $n_{m,n}$, and compute the mean and the variance respectively as
\begin{eqnarray}
\mathbb{E}\{n_{m,n}\}&\approx& 0 \label{eq:mean} \nonumber\\
\text{and} \quad \mathbb{E}\{n_{m,n}^2\} &\approx& 4{r}^{2}_{m,n}\sigma^2_{m,n}.\nonumber
\end{eqnarray}

Since all the sensors are mounted on the rigid body, it is reasonable to assume that the noise from an anchor to any sensor (and, hence to the rigid body) is approximately the same, especially when the anchors are far away from the rigid body. Hence, we use a simplified noise model\footnotemark[1]  with variance 
\begin{equation}
\begin{aligned}
\label{eq:variance_approx2}
\mathbb{E}\{n_{m,n}^2\} \approx {r}^{2}_{m,1}\sigma^2_{m,1} = {r}^{4}_{m,1} / \zeta= \sigma^2_{m}
\end{aligned}
\end{equation} 
where we choose sensor ${\bf s}_1$ arbitrarily just for illustration purposes, and in principle, this can be any sensor on the rigid body. 

\footnotetext[1]{More accurate noise models could be considered, but this is not the main focus of this paper.}

The problem discussed in this paper can now briefly be stated as follows: Given the range measurements between each sensor-anchor pair and the topology of the sensors on a rigid body, jointly determine the position and orientation (rotation along each dimension) of the rigid body in $\mathbb{R}^3$.
\section{Preliminaries} \label{sec:prel}
Defining 
\begin{eqnarray}
{\bf d}_n &=& [d_{1,n},d_{2,n},\ldots,d_{M,n}]^T \in \mathbb{R}^{M \times 1},\nonumber\\
\text{and} \quad {\bf u} &=& [\|{\bf a}_1\|^2,\|{\bf a}_2\|^2,\ldots,\|{\bf a}_M\|^2]^T \in \mathbb{R}^{M \times 1},\nonumber
 \end{eqnarray} 
we can write the squared-range measurements between the $n$th sensor to each of the anchors in vector form as
\begin{equation}
\label{eq:sqpairwisedistance}
{\bf d}_n = {\bf u} - 2{\bf A}^T{\bf s}_n + \|{\bf s}_n\|^2 {\bf 1}_M + {\bf n}_n,
\end{equation} where ${\bf n}_n = [n_{1,n},n_{2,n},\ldots,n_{M,n}]^T \in \mathbb{R}^{M \times 1}$ is the error vector. The covariance matrix of the error vector will be 
\begin{equation*}
{\boldsymbol \Sigma}_{\bf n} = \mathbb{E}\{{\bf n}_n{\bf n}_n^T\}= \mathrm{diag}(\sigma_1^2,\sigma_2^2,\ldots,\sigma_M^2) \in \mathbb{R}^{M \times M}.
\end{equation*} 
Let us now pre-whiten (\ref{eq:sqpairwisedistance}) to obtain an identity noise covariance matrix by multiplying both sides of (\ref{eq:sqpairwisedistance}) with a pre-whitening matrix ${\bf W} \in \mathbb{R}^{M \times M}$, which leads to
\begin{equation}
\label{eq:sqpairwisedistance1}
{\bf W} {\bf d}_n = {\bf W} ({\bf u} - 2{\bf A}^T{\bf s}_n + \|{\bf s}_n\|^2 {\bf 1}_M + {\bf n}_n).
\end{equation} 
The optimal ${\bf W}$ is ${\bf W}^* = \boldsymbol{\Sigma}_{\bf n}^{-1/2}$, which however, depends on the unknown parameter $r_{m,1}$. 
Hence, we use ${\bf W} = \hat{\boldsymbol{\Sigma}}_{\bf n}^{-1/2}$, where $\hat{\boldsymbol{\Sigma}}_{\bf n}$ is the estimated noise covariance matrix computed using $\hat{\sigma}^2_m =  {d}^{2}_{m,1}/\zeta$, which is  based on the measured parameter $d_{m,1}$.


We now try to eliminate $\|{\bf s}_n\|^2$ in \eqref{eq:sqpairwisedistance1}, which can be done by projecting out the vector ${\bf W}{\bf 1}_M$. For this, we apply an orthogonal projection matrix
\begin{equation*}
\label{eq:PM}
{\bf P}_M \triangleq {\bf I}_M - \frac{{\bf W}{\bf 1}_M{\bf 1}_M^T{\bf W}}{{\bf 1}_M^T{\bf W}{\bf W}{\bf 1}_M} \quad \in \mathbb{R}^{M \times M}, 
\end{equation*}
such that ${\bf P}_M{\bf W}{\bf 1}_M ={\bf 0}$. However, since this would again color the noise, we propose to use an isometry decomposition of ${\bf P}_M$, i.e.,  
\begin{equation*}
\label{eq:PMdecomposition}
{\bf P}_M = {\bf U}_M{\bf U}_M^T,
\end{equation*}
where ${\bf U}_M$ is an $M \times (M-1)$ matrix obtained by collecting orthonormal basis vectors of the null-space of ${\bf W}{\bf 1}_M$ so that ${\bf U}_M^T{\bf W}{\bf 1}_M={\bf 0}_{M-1}$.
Then, in order to eliminate the $\|{\bf s}_n\|^2{\bf W}{\bf 1}_M$ term in (\ref{eq:sqpairwisedistance1}) {\it without} coloring the noise, we left-multiply both sides of (\ref{eq:sqpairwisedistance1}) with ${\bf U}_M^T$, which leads to 
 \begin{equation}
{\bf U}_M^T{\bf W} ({\bf d}_n-{\bf u}) =-2{\bf U}_M^T{\bf W}{\bf A}^T{\bf s}_n + {\bf U}_M^T{\bf W}{\bf n}_n.
\label{eq:linear}
\end{equation} 
Stacking (\ref{eq:linear}) for all the $N$ sensors on the rigid body, we obtain
\begin{equation}
{\bf U}_M^T {\bf W} {\bf D}= -2{\bf U}_M^T{\bf W}{\bf A}^T{\bf S}  + {\bf U}_M^T {\bf W} {\bf N} \label{eq:linear1},
\end{equation}
where 
\begin{eqnarray}
\label{eq:D}
{\bf D} &=& [ {\bf d}_1,{\bf d}_2,\ldots, {\bf d}_N ] - {\bf u} {\bf 1}_N^T  \in \mathbb{R}^{M \times N},\nonumber\\
\text{and} \quad {\bf N} &=& [{\bf n}_1,{\bf n}_2,\ldots,{\bf n}_N]  \in \mathbb{R}^{M \times N}.\nonumber
\end{eqnarray}
Note that the approximation of the noise model in (\ref{eq:variance_approx2}) allows this stacking by using a common pre-whitening matrix ${\bf W}$ for all the sensors. 

\subsection{Classical LS-based localization}
The pre-whitened linear model in \eqref{eq:linear1} can be further simplified to
\begin{equation}
\label{eq:locallinearmodel}
\bar{\bf D}= \bar{\bf A}{\bf S} + \bar{\bf N}, 
\end{equation}
where we have introduced the following matrices:
\begin{equation*}
\label{eq:subs}
\begin{aligned}
\bar{\bf D} &= {\bf U}_M^T {\bf W}{\bf D} \in \mathbb{R}^{(M-1) \times N},  \\
\bar{\bf A} &= -2{\bf U}_M^T{\bf W}{\bf A}^T \in \mathbb{R}^{(M-1) \times 3},\\
\text{and} \quad \bar{\bf N} &= {\bf U}_M^T {\bf W}{\bf N} \in \mathbb{R}^{(M-1) \times N}.
\end{aligned}
\end{equation*}

Since \eqref{eq:locallinearmodel} is row-wise white, we can use the classical (unweighted) LS solution to estimate the absolute position of the sensors as
\begin{equation}
\label{eq:LSlocalization}
\begin{aligned}
\hat{\bf S}_{LS}  &=\argmin_{\bf S} \quad {\|\bar{\bf D}-\bar{\bf A}{\bf S}\|}_F^2\\
&= \bar{\bf A}^\dag \bar{\bf D},
\end{aligned}
\end{equation} 
which is unique if $\bar{\bf A}$ is full column-rank, and this requires $M \geq 4$.

Note that in this classical LS-based localization, the knowledge about the known sensor topology is not exploited, and the absolute position of each sensor is estimated separately.
 
\subsection{Known sensor topology and the Stiefel manifold}

A  Stiefel manifold~\cite{Stiefel1999} in three dimensions, commonly denoted by $\mathcal{V}_{3,3}$, is the set of all $3 \times 3$ unitary matrices 
${\bf Q} = [{\bf q}_1,{\bf q}_2,{\bf q}_3] \in \mathbb{R}^{3 \times 3}$, i.e., 
\begin{equation}
\label{eq:stiefel}
\mathcal{V}_{3,3}= \{{\bf Q} \in \mathbb{R}^{3 \times 3} : {\bf Q}^T{\bf Q}= {\bf Q}{\bf Q}^T = {\bf I}_3\}.
\end{equation} 

The absolute position of the $n$th sensor can be written as an affine function of a point on the Stiefel manifold, i.e.,
\begin{eqnarray}
{\bf s}_n &=&  c_{n,1}{\bf q}_1 + c_{n,2}{\bf q}_2 + c_{n,3}{\bf q}_3 + {\bf t}\nonumber\\
&=&  {\bf Q} {\bf c}_n + {\bf t}\label{eq:sensorstief},
\end{eqnarray}
where ${\bf t} \in \mathbb{R}^{3 \times 1}$ denotes the translation and is unknown. 

More specifically, the parameter vector ${\bf t}$ refers to the unknown position of the rigid body. The combining weights ${\bf c}_n$ are equal to the {\it known} coordinates of the $n$th sensor in the reference frame, as introduced in Section \ref{sec:prel}. This means that the unknown unitary matrix ${\bf Q}$ actually tells us how the rigid body has rotated in the reference frame. When there is no rotation, then ${\bf Q} = {\bf I}_3$. The relation in \eqref{eq:sensorstief} is sometimes also referred to as the {\it rigid body transformation}. The rotation matrices can uniquely (both geometrically and kinematically) represent the orientation of a rigid body unlike Euler angles or unit quaternions (see~\cite{rigidbodyControlmag} for more details). The rigid body transformation is also used in computer vision applications for motion parameter estimation~\cite{ArunCTLS}.

If we define ${\bf C} = [{\bf c}_1,{\bf c}_2,\ldots, {\bf c}_N]$, then as in \eqref{eq:sensorstief}, the absolute position of all the sensors (or the sensor array) can be written as an affine function of the Stiefel manifold
\begin{equation}
{\bf S}  =  {\bf Q}{\bf C} + {\bf t}{\bf 1}_N^T= \overbrace{\left[\begin{array}{c|c}{\bf Q} & {\bf t}\end{array}\right]}^{{\bf Q}_e} \overbrace{\left[\begin{array}{c}{\bf C} \\ \hline {\bf 1}^T_N \end{array}\right]}^{{\bf C}_e}.\label{eq:sensorstief1}
\end{equation}
In (\ref{eq:sensorstief1}), we express the unknown sensor locations ${\bf S}$ in terms of the unknown rotations ${\bf Q}$, an unknown translation ${\bf t}$, and a known sensor topology ${\bf C}$. 
\section{Proposed estimators: known topology} \label{sec:LS}
In this section, we propose a number of algorithms to estimate the position of the rigid body, i.e., ${\bf t}$, and the orientation of the rigid body, i.e., ${\bf Q}$. To start, we propose an LS-based estimator to jointly estimate ${\bf Q}$ and ${\bf t}$. 
\subsection{LS estimator (Unconstrained)} 

Substituting (\ref{eq:sensorstief1}) in  (\ref{eq:linear1}) we arrive at the following linear model
\begin{equation*}
\label{eq:jointLS}
{\bf U}_M^T {\bf W} {\bf D}= -2{\bf U}_M^T{\bf W}{\bf A}^T{{\bf Q}_e}{{\bf C}_e}  + {\bf U}_M^T {\bf W}{\bf N} 
\end{equation*}
which can be written as
\begin{equation}
\label{eq:jointLS1}
\bar{\bf D}= \bar{\bf A}{{\bf Q}_e}{{\bf C}_e}  + \bar{\bf N}
\end{equation}
recalling that $\bar{\bf D} = {\bf U}_M^T {\bf W}{\bf D}$, $\bar{\bf A} = -2{\bf U}_M^T{\bf W}{\bf A}^T $, and $\bar{\bf N}={\bf U}_M^T {\bf W}{\bf N}$ as defined earlier. 
Using the matrix property
\begin{equation*}
\mathrm{vec}({\bf A}{\bf B}{\bf C}) = ({\bf C}^T \otimes {\bf A}) \mathrm{vec}({\bf B}),
\end{equation*}
we can vectorize (\ref{eq:jointLS1}), leading to
\begin{equation}
\label{eq:jointLS2}
\bar{\bf d}= ({\bf C}_e^T \otimes \bar{\bf A}) {{\bf q}_e}  + \bar{\bf n},
\end{equation}
where 
\begin{eqnarray*}
{\bf q}_e = \mathrm{vec}({\bf Q}_e) &=& [{\bf q}_1^T, {\bf q}_2^T, {\bf q}_3^T, {\bf t}^T]^T \in \mathbb{R}^{12 \times 1}, \\
\bar{\bf d} &=& \mathrm{vec}(\bar{\bf D}) \in \mathbb{R}^{(M-1)N \times 1},\\
\text{and} \quad \bar{\bf n} &=& \mathrm{vec}(\bar{\bf N}) \in \mathbb{R}^{(M-1)N \times 1}.
\end{eqnarray*}

\begin{mylem}
The covariance matrix of $\bar{\bf n}$ will be $\mathbb{E}\{\bar{\bf n}\bar{\bf n}^T\}  \approx {\bf I}_{(M-1)N}$.
\end{mylem}
\begin{IEEEproof}
See Appendix~\ref{app:covmatrix}.
\end{IEEEproof}

Due to the whiteness of \eqref{eq:jointLS2}, as shown by the lemma, we propose to jointly estimate the unknown rotations ${\bf Q}$ and the translation ${\bf t}$ using the following (unweighted) LS estimator
\begin{equation}
\begin{aligned}
\label{eq:jointLSsolution}
{\hat{\mathbf q}}_{e,LS} &= \argmin_{{\bf q}_e} \quad {\|\bar{\bf d} - ({\bf C}_e^T \otimes \bar{\bf A}) {{\bf q}_e}\|}_2^2 \\
&= ({\bf C}_e^T \otimes \bar{\bf A})^\dag \bar{\bf d},
\end{aligned}
\end{equation} which will have a unique solution if ${\bf C}_e^T \otimes \bar{\bf A}$ has full column-rank, i.e., ${\bf C}_e^T$ and $\bar{\bf A}$ are both full-column rank, and this requires $(M-1)N \geq 12$. Finally, we have 
\begin{equation}
\label{eq:jointLSsolution1}
\hat{\mathbf Q}_{e,LS} = \mathrm{vec}^{-1}(\hat{\mathbf q}_{e,LS})=\left[\begin{array}{c|c}{\hat{\bf Q}}_{LS} & \hat{\bf t}_{LS}\end{array}\right].
\end{equation} 
\subsection{Unitarily constrained LS  (UC-LS) estimator}
The solution of the unconstrained LS estimator \eqref{eq:jointLSsolution1} does not necessarily lie in the set $\mathcal{V}_{3,3}$, i.e., the columns of the LS estimate $\hat{\bf Q}_{LS}$ obtained in (\ref{eq:jointLSsolution}) are generally not orthogonal to each other and they need not have a unit norm.  Hence,  we next propose two LS estimators with a unitary constraint on ${\bf Q}$. Both these estimators solve an optimization problem on the Stiefel manifold.

For this purpose, we decouple the rotations and the translations in (\ref{eq:sensorstief1}) by eliminating the vector ${\bf 1}^T_N$, and hence the matrix ${\bf t}{\bf 1}_N^T$. In order to eliminate ${\bf t}{\bf 1}_N^T$, we use an isometry matrix ${\bf U}_N$, and as earlier this matrix is obtained by the isometry decomposition of ${\bf P}_N = {\bf I}_N - \frac{1}{N} {\bf 1}_N{\bf 1}_N^T$, given by
\begin{eqnarray}
{\bf P}_N &=& {\bf U}_N{\bf U}_N^T, 
\end{eqnarray}
where ${\bf U}_N$ is an $N \times (N-1)$ matrix obtained by collecting orthonormal basis vectors of the null-space of ${\bf 1}_N$ such that ${\bf 1}_N^T{\bf U}_N= {\bf 0}_{N-1}^T$.
Right-multiplying ${\bf U}_N$ on both sides of (\ref{eq:sensorstief1}) leads to
\begin{equation}
{\bf S}{\bf U}_N  =  {\bf Q}{\bf C}{\bf U}_N .\label{eq:sensorstief2}
\end{equation}
Combining (\ref{eq:linear1}) and (\ref{eq:sensorstief2}) we get the following linear model 
\begin{equation*}
{\bf U}_M^T {\bf W}{\bf D}{\bf U}_N= \bar{\bf A}{\bf Q}{\bf C}{\bf U}_N+ {\bf U}_M^T {\bf W}{\bf N}{\bf U}_N
\end{equation*}
which can be further simplified as
\begin{equation}
\begin{aligned}
\label{eq:model1}
\tilde{\bf D}= \bar{\bf A}{\bf Q}\bar{\bf C}+ \tilde {\bf N}
\end{aligned}
\quad \Leftrightarrow \quad
\begin{aligned}
\tilde{\bf d}= (\bar{\bf C}^T \otimes \bar{\bf A}){\bf q}+ \tilde {\bf n},
\end{aligned}
\end{equation}
where $\tilde{\bf d} = \mathrm{vec}(\tilde{\bf D})$, ${\bf q} = \mathrm{vec}({\bf Q})$, and $\tilde{\bf n}=\mathrm{vec}(\tilde{\bf N})$. Here, we have introduced the following matrices:
\begin{eqnarray*}
\tilde{\bf D}&=& {\bf U}_M^T{\bf W} {\bf D}{\bf U}_N \in \mathbb{R}^{(M-1) \times (N-1)},\\
\bar{\bf C} &=& {\bf C} {\bf U}_N \in \mathbb{R}^{3 \times (N-1)},\\
\tilde{\bf N} &=& {\bf U}_M^T {\bf W}{\bf N}{\bf U}_N \in \mathbb{R}^{(M-1) \times (N-1)}.
\end{eqnarray*}
\begin{mylem}
The covariance matrix of $\tilde{\bf n}$ will be $\mathbb{E}\{\tilde{\bf n}\tilde{\bf n}^T\}  \approx {\bf I}_{K}$, with $K=(M-1)(N-1)$.
\end{mylem}
\begin{IEEEproof}
See Appendix~\ref{app:covmatrix}.
\end{IEEEproof}

Due to the whiteness of \eqref{eq:model1}, as shown by the lemma, we will try to estimate ${\bf Q}$ based on an (unweighted) LS problem with a quadratic equality constraint, as given by
\begin{equation}
\label{eq:UCLS}
\begin{aligned}
&\argmin_{{\bf Q}}  \, {\| \tilde{\bf d} - (\bar{\bf C}^T \otimes \bar{\bf A}){\bf q} \|}_2^2 \\
& \quad s.t. \quad {\bf Q}^T{\bf Q} = {\bf I}_3.
\end{aligned}
\end{equation}

The optimization problem in (\ref{eq:UCLS}) is non-convex due to the quadratic equality constraint, and does not have a closed form analytical solution. However, such optimization problems can be solved iteratively as will be discussed later on. Alternatively, the optimization problem in (\ref{eq:UCLS}) can be simplified and brought to the standard form of a {\it orthogonal Procrustes problem} (OPP).

\subsubsection{Simplified UC-LS (SUC-LS)}
Using $\check{\bf D} \triangleq \bar{\bf A}^\dag \tilde{\bf D}$, the simplified unitarily constrained-LS problem is then given as
\begin{equation}
 \label{eq:CLSprob}
  \begin{aligned}
&\argmin_{{\bf Q}} \quad {\|{\bf Q}\bar{\bf C} - \check{\bf D} \|}_F^2 \\
 &\quad s.t. \quad {\bf Q}^T{\bf Q} = {\bf I}_3
\end{aligned}
\end{equation}where we assume that $\bar{\bf A}$ has full column-rank.

This optimization problem is commonly referred to as the {\it orthogonal Procrustes problem} (OPP), and is generally used to compute the rotations between subspaces.

\begin{myrem} [Anchor placement]
For $M \geq 3$, the anchor positions can be designed such that the matrix $\bar{\bf A}$ will be full column-rank and well-conditioned (see e.g.~\cite{anchorchepuri}). Then, the matrix $\bar{\bf A}$ is left-invertible, i.e., $\bar{\bf A}^\dag \bar{\bf A} = {\bf I}_3$. 
\end{myrem}

 \begin{mytheo}[Solution to SUC-LS problem] 
The constrained LS problem in (\ref{eq:CLSprob}) has a closed-form analytical solution given by $\hat{\bf Q}_{SUC-LS} = {\bf V}{\bf U}^T$, where ${\bf U}$ and ${\bf V}$ are obtained from the singular value decomposition (SVD) of $\bar{\bf C}\check{\bf D}^T$ which is given by ${\bf U}{\bf \Sigma}{\bf V}^T$. The obtained solution is unique, if and only if $\bar{\bf C}\check{\bf D}^T$ is non-singular.
\end{mytheo}
\begin{proof}
See~\cite[pg. 601]{golub1996matrix}.
\end{proof} 
Subsequently, the SUC-LS estimate of the translation ${\bf t}$ can be computed using $\hat{\bf Q}_{SUC-LS}$ in \eqref{eq:sensorstief1} and \eqref{eq:locallinearmodel}, i.e.,
\begin{equation}
\begin{aligned}
\hat{\bf t}_{SUC-LS}  &= \min_{\bf t} \quad {\| \bar{\bf D} -\bar{\bf A}(\hat{\bf Q}_{SUC-LS}{\bf C} + {\bf t}{\bf 1}_N^T) \|}_F^2\\
&= \frac{1}{N}(\bar{\bf A}^\dag\bar{\bf D}-\hat{\bf Q}_{CLS}{\bf C}){\bf 1}_N.
\end{aligned}
\end{equation}

\subsubsection{Optimal unitarily constrained LS (OUC-LS) estimator}

Pseudo inverting $\bar{\bf A}$ in (\ref{eq:CLSprob}) colors the noise which makes the unweighted LS problem in~\eqref{eq:CLSprob} suboptimal. This can be avoided by solving the OUC-LS formulation that was introduced earlier, which is given by
\begin{equation}
\label{eq:iterMLprob}
\begin{aligned}
\hat{\bf Q}_{OUC-LS} =& \argmin_{{\bf Q}}  \, {\| \tilde{\bf d} - (\bar{\bf C}^T \otimes \bar{\bf A}){\bf q} \|}_2^2 \\
& \quad s.t. \quad {\bf Q}^T{\bf Q} = {\bf I}_3.
\end{aligned}
\end{equation}
This is a linear LS problem on the Stiefel manifold which can be written as
\begin{equation}
\label{eq:iterMLprobgeneral}
\begin{aligned}
\hat{\bf Q}_{OUC-LS} &= \quad \argmin_{\bf Q} \, {\| f({\bf Q}) - {\bf b}\|}_2^2\\
&\quad \,s.t. \quad {\bf Q} \in \mathcal{V}_{3,3} \, ,
\end{aligned}
\end{equation}
with $f({\bf Q}): \mathbb{R}^{K \times 9} {\rightarrow} \mathbb{R}^K$ being a linear function in ${\bf Q}$, and for  \eqref{eq:iterMLprob} we use
\begin{equation}
\begin{aligned}
f({\bf Q}) &:= (\bar{\bf C}^T \otimes \bar{\bf A}) \mathrm{vec}({\bf Q}) \in \mathbb{R}^{K \times 1}\\
\text{and} \qquad {\bf b} &:= \tilde{\bf d}=\mathrm{vec}(\tilde{\bf D})\in \mathbb{R}^{K \times 1}.
\end{aligned}
\end{equation}

The optimization problem in \eqref{eq:iterMLprob} is a generalization of the OPP, and is sometimes also referred to as the {\it weighted orthogonal Procrustes problem} (WOPP)~\cite{WOPP}. Unlike the OPP of \eqref{eq:CLSprob}, which has a closed-form analytical solution, the optimization problem \eqref{eq:iterMLprob} does not have a closed-form solution. However, it can be solved using iterative methods based on either Newton's method~\cite{WOPP} or steepest descent~\cite{steepest} (sometimes also combinations of these two methods). Note that such algorithms can also be used for finding unitary matrices in joint diagonalization problems (e.g., in blind beamforming and blind source separation~\cite{allejanCMA,steepest}).

The advantages and disadvantages of both Newton's and steepest descent based algorithms are well-known (see~\cite{Boyd}).  In this paper, we restrict ourselves to Newton's method for solving \eqref{eq:iterMLprob} because of the availability of a good built-in initial value for the iterative algorithm, and because of its quadratic convergence. For self-consistency purposes, the algorithm is briefly described in Appendix~\ref{app:newton}. The algorithm from~\cite{WOPP} based on Newton's method is adapted to suit our problem, and it is summarized as Algorithm~\ref{alg:newton}.  
\begin{algorithm}[!t]
\caption{OUC-LS based on Newton's method}
\label{alg:newton}
\begin{algorithmic}
\item[1.] \textbf{Compute} the initial value ${\bf Q}_0$ by solving \eqref{eq:init1} and \eqref{eq:init2}.
\item[2.] \textbf{initialize} $i=0$, $\epsilon=10^{-6}$, $\epsilon_0 = \epsilon +1$.
\item[3.] \textbf{while} $\epsilon_i > \epsilon$
\item[4.] \hspace*{6mm} \textbf{If} $({\bf J}_Q^T{\bf J}_Q + {\bf H}) \succ {\bf 0}$
\item[5.] \hspace*{10mm} \textbf{compute} a {\it Newton} step ${\bf x}_N$ using \eqref{eq:Newton}.
\item[6.] \hspace*{6mm} \textbf{else} 
\item[7.] \hspace*{10mm} \textbf{compute} a {\it Gauss-Newton} step ${\bf x}_{GN}$ using \eqref{eq:GaussNewton}.
\item[8.] \hspace*{6mm} \textbf{compute} the optimal step-length $\hat{\gamma}$ using \eqref{eq:optimalsteplength}.
\item[9.] \hspace*{6mm} \textbf{update} ${\bf Q}_{i+1} = {\bf Q}_i \exp({\bf X}(\hat{\gamma}{\bf x}))$.
\item[10.] \hspace*{4.5mm} \textbf{increment} $i = i+1$.
\item[11.] \hspace*{4.5mm} \textbf{compute} $\epsilon_{i+1} = \frac{{\|{\bf J}_Q^T(f({\bf Q}_j)-{\bf b})\|}_2}{{\|{\bf J}_Q\|}_F{\|f({\bf Q}_j)-{\bf b}\|}_2}$.
\item[12.] \textbf{end while}.
\end{algorithmic}
\end{algorithm} 
Note that the algorithm does not converge to an optimal solution if the solution from SUC-LS is used as an initial value for the Newton's method due to the inverse operation in SUC-LS. In addition, as observed during the simulations, the iterative algorithm converges very quickly (less than 5 iterations). The readers are further referred to~\cite{WOPP} for a more profound treatment, and a performance analysis of the iterative algorithm.

As earlier, the estimate for the translation ${\bf t}$ can then be computed using $\hat{\bf Q}_{OUC-LS}$, and is given by 
\begin{equation}
\begin{aligned}
\hat{\bf t}_{OUC-LS} &= \frac{1}{N}(\bar{\bf A}^\dag\bar{\bf D}-\hat{\bf Q}_{OUC-LS}{\bf C}){\bf 1}_N.
\end{aligned}
\end{equation}

\subsection{Topology-aware (TA) localization}\label{sec:toploc}

A complementary by-product of the rigid body localization is the {\it topology-aware} localization. In this case, the position and orientation estimation is not the main interest, but the absolute position of each sensor node has to be estimated, given that the sensors lie on a certain manifold (or follow a certain topology). This latter information can be used as a constraint for estimating the sensor positions rather than estimating it separately. For the rigid body constraint, using $\hat{\bf Q}$ and $\hat{\bf t}$ obtained from either SUC-LS or OUC-LS estimator, we can compute the absolute positions of each sensor on the rigid body as
\begin{equation}
\hat{\bf S}_{TA} = \hat{\bf Q}{\bf C} + \hat{\bf t}{\bf 1}_N^T.
\end{equation}

\section{Perturbations on the known topology} \label{sec:TLS}
In the previous section, we assumed that the position of the sensors in the {\it reference frame} on a rigid body, i.e., the matrix ${\bf C}$, is accurately known. In practice, there is no reason to believe that errors are restricted only to the range measurements and there are no perturbations on the initial sensor positions. Such perturbations can be introduced for instance during  fabrication or if the body is not entirely rigid. 

So let us now assume that the position of the $n$th sensor in the reference frame ${\bf c}_n$ is noisy, and let us denote the perturbation on ${\bf c}_n$ by ${\bf e}_n$, and the perturbation on ${\bf C} =  [{\bf c}_1,{\bf c}_2,\ldots,{\bf c}_N]$ by ${\bf E} = [{\bf e}_1,{\bf e}_2,\ldots,{\bf e}_N]$. To account for such errors in the model, we propose total-least-squares (TLS) estimates for \eqref{eq:CLSprob} and \eqref{eq:iterMLprob}, again with unitary constraints.

\subsection{Simplified unitarily constrained TLS estimator (SUC-TLS)}

Taking the perturbations on the known topology into account, the data model in (\ref{eq:model1}) will be modified as 
\begin{equation}
\label{eq:TLSdatamodel}
{\bf Q}(\bar{\bf C} + \bar{\bf E})  = \check{\bf D} +  \check{\bf  N}
\end{equation}
where $\bar{\bf  E} ={\bf E}{\bf U}_N$ and $\check{\bf N}= \bar{\bf  A}^\dag \tilde{\bf  N}$.

The solution to the data model in (\ref{eq:TLSdatamodel}) leads to the classical TLS optimization problem, but now with a unitary constraint. The SUC-TLS optimization problem is given by
\begin{equation}
 \label{eq:TLSoptproblem}
 \begin{aligned}
&\argmin_{{\bf Q}} \quad {\|\bar{\bf E}\|}^2_F + {\|\check{\bf  N}\|}^2_F, \\
&s.t. \quad  {\bf Q}(\bar{\bf C} + \bar{\bf E})  = \check{\bf D} +  \check{\bf  N} \quad \text{and} \quad {\bf Q}^T{\bf Q} = {\bf I}_3.
\end{aligned}
\end{equation} 

 \begin{mytheo}[Solution to SUC-TLS~\cite{ArunCTLS}] 
The SUC-TLS problem in (\ref{eq:TLSoptproblem}) has the same solution as the simplified unitarily constrained LS problem.
\end{mytheo}
\begin{proof}
For any ${\bf Q}$, we can re-write the constraint in (\ref{eq:TLSoptproblem}) as
\begin{equation*}
 \label{eq:TLSproof1}
 \begin{aligned}
 \left[\begin{array}{c|c}{\bf Q} & -{\bf I}\end{array}\right]  \left[\begin{array}{c}\bar{\bf E} \\ \check{\bf N}\end{array}\right] = -  \left[\begin{array}{c|c}{\bf Q} & -{\bf I}\end{array}\right]  \left[\begin{array}{c}\bar{\bf C} \\ \check{\bf D}\end{array}\right].
\end{aligned}
\end{equation*}
Using the unitary constraint on ${\bf Q}$, and right-inverting the wide matrix $\left[\begin{array}{c|c}{\bf Q} & -{\bf I}\end{array}\right]$ we get
\begin{equation}
 \label{eq:TLSproof1}
 \begin{aligned}
\left[\begin{array}{c}\bar{\bf E} \\ \check{\bf N}\end{array}\right] &= - \frac{1}{2}  \left[\begin{array}{c}{\bf Q}^T \\ -{\bf I}\end{array}\right]   \left[\begin{array}{c|c}{\bf Q} & -{\bf I}\end{array}\right]  \left[\begin{array}{c}\bar{\bf C} \\ \check{\bf D}\end{array}\right] \\
&= - \frac{1}{2} \left[\begin{array}{cc}{\bf I} &-{\bf Q}^T \\-{\bf Q} & {\bf I}\end{array}\right] \left[\begin{array}{c}\bar{\bf C} \\ \check{\bf D}\end{array}\right] \\ &= -\frac{1}{2}\left[\begin{array}{c}\bar{\bf C}-{\bf Q}^T\check{\bf D} \\ \check{\bf D} - {\bf Q}\bar{\bf C}\end{array}\right]
\end{aligned}
\end{equation}
We can now re-write the objective in (\ref{eq:TLSoptproblem}) to compute the minimum-norm square solution as 
\begin{equation}
 \label{eq:TLSproof2}
 \begin{aligned}
&\mathrm{tr}\left(\left[\begin{array}{c|c}\bar{\bf E}^T & \check{\bf N}^T\end{array}\right] \left[\begin{array}{c}\bar{\bf E} \\ \check{\bf N}\end{array}\right]\right)\\
 &= \mathrm{tr}(\frac{1}{2} (\bar{\bf C}^T\bar{\bf C} - \check{\bf D}^T{\bf Q}\bar{\bf C} - \bar{\bf C}^T{\bf Q}^T\check{\bf D} + \check{\bf D}^T\check{\bf D})) \\
 &= \frac{1}{2}{\|\bar{\bf C}\|}_F^2 - \mathrm{tr}({\bf Q}\bar{\bf C}\check{\bf D}^T) +  \frac{1}{2}{\|\check{\bf D}\|}_F^2.
\end{aligned}
\end{equation}  
The solution to the UC-TLS problem is then obtained by optimizing the term depending only on ${\bf Q}$, i.e., by maximizing $\mathrm{tr}({\bf Q}\bar{\bf C}\check{\bf D}^T)$. This is the same cost  as that of the SUC-LS problem in \eqref{eq:CLSprob}. Hence, the solution to the unitarily constrained TLS problem is 
\begin{equation}
 \label{eq:CTLS}
 \begin{aligned}
\hat{\bf Q}_{SUC-TLS} =  \hat{\bf Q}_{SUC-LS}= {\bf V} {\bf U}^T
\end{aligned}
\end{equation}
where the matrices ${\bf U}$ and ${\bf V}$ are again obtained by computing the SVD of $\bar{\bf C}\check{\bf D}^T: \bar{\bf C}\check{\bf D}^T= {\bf U}{\bf \Sigma}{\bf V}^T$.
\end{proof}

The algorithms to compute the SUC-LS and SUC-TLS estimators are summarized as Algorithm \ref{alg:CTLS}.
\begin{algorithm}[!t]
\caption{Summary of SUC-LS or SUC-TLS estimators}
\label{alg:CTLS}
\begin{algorithmic}
\item[1.] \textbf{Given} $\bar{\bf C}$ and measurements $\check{\bf D}$. 
\item[2.] \textbf{compute} $\bar{\bf C}\check{\bf D}^T$. \\
\item[3.] \textbf{compute} SVD of $\bar{\bf C}\check{\bf D}^T$:  $\bar{\bf C}\check{\bf D}^T={\bf U}{\bf \Sigma}{\bf V}^T$.
\item[4.] $\hat{\bf Q}_{SUC-LS}= \hat{\bf Q}_{SUC-TLS} = {\bf V} {\bf U}^T$.
\item[5.] $\hat{\bf t}_{SUC-LS}= \hat{\bf t}_{SUC-TLS} = \frac{1}{N}(\bar{\bf A}^\dag\bar{\bf D}-\hat{\bf Q}_{SUC-LS}{\bf C}){\bf 1}_N$.
\end{algorithmic}
\end{algorithm} 


\subsection{Optimal unitarily constrained TLS estimator (OUC-TLS)}

Similar to the OUC-LS formulation, the TLS estimator can be derived without pseudo-inverting the matrix $\bar{\bf A}$ in \eqref{eq:TLSdatamodel}. The data model taking into account the error in the known sensor topology is then given by
\begin{equation}
\bar{\bf A}{\bf Q}(\bar{\bf C} + \bar{\bf E})  = \tilde{\bf D} +  \tilde{\bf  N}.\label{eq:iterTLSmodel}
\end{equation}

The optimal unitarily constrained TLS (OUC-TLS) optimization problem is given by
\begin{equation}
 \label{eq:optCTLSprob}
 \begin{aligned}
& \argmin_{{\bf Q}} \quad {\|\bar{\bf E}\|}^2_F + {\|\tilde{\bf  N}\|}^2_F, \\
&s.t. \quad  \bar{\bf A}{\bf Q}(\bar{\bf C} + \bar{\bf E})  = \tilde{\bf D} +  \tilde{\bf  N}, \quad \text{and} \quad {\bf Q}^T{\bf Q} = {\bf I}_3.
\end{aligned}
\end{equation} 

 \begin{mytheo}[Solution to OUC-TLS] 
The optimal unitarily constrained TLS problem in \eqref{eq:optCTLSprob} has the same solution as a specifically weighted OUC-LS, i.e., it is the solution to
\begin{equation}
\label{eq:optCTLSprob1}
\begin{aligned}
{\hat{\bf Q}}_{OUC-TLS} =& \; \argmin_{{\bf Q}}  \, {\| {\boldsymbol \Lambda}^{-1/2}(\bar{\bf A}{\bf Q}\bar{\bf C} - \tilde{\bf D}) \|}_F^2 \\
&s.t. \quad {\bf Q}^T{\bf Q} = {\bf I}_3
\end{aligned}
\end{equation} 
where ${\boldsymbol \Lambda} = (\bar{\bf A}\bar{\bf A}^T + {\bf I}_{M-1}) \in \mathbb{R}^{(M-1) \times (M-1)}$ is a weighting matrix.
\end{mytheo}

\begin{IEEEproof}
For any ${\bf Q}$ the constraint in the optimization problem \eqref{eq:optCTLSprob} can be written as
\begin{equation}
 \label{eq:optCTLSproof1}
 \begin{aligned}
 \left[\begin{array}{c|c}\bar{\bf A}{\bf Q} & -{\bf I}\end{array}\right]  \left[\begin{array}{c}\bar{\bf E} \\ \tilde{\bf N}\end{array}\right] = -  \left[\begin{array}{c|c}\bar{\bf A}{\bf Q} & -{\bf I}\end{array}\right]  \left[\begin{array}{c}\bar{\bf C} \\ \tilde{\bf D}\end{array}\right]. \\ 
\end{aligned}
\end{equation}
Multiplying both sides of \eqref{eq:optCTLSproof1} with the right-inverse of the wide-matrix $[\bar{\bf A}{\bf Q} | -{\bf I}]$ given by
\begin{equation}
 \label{eq:optCTLSproof2}
 \begin{aligned}
\left[\begin{array}{c|c}\bar{\bf A}{\bf Q} & -{\bf I}\end{array}\right]^{\dag}  = \left[\begin{array}{c}{\bf Q}\bar{\bf A}^T \\ -{\bf I}\end{array}\right] {(\bar{\bf A}\bar{\bf A}^T + {\bf I}_{M-1})}^{-1},
\end{aligned} 
\end{equation}
we get
\begin{equation}
 \label{eq:optCTLSproof3}
 \begin{aligned}
 \left[\begin{array}{c}\bar{\bf E} \\ \tilde{\bf N}\end{array}\right] = -  \left[\begin{array}{c}{\bf Q}\bar{\bf A}^T \\ -{\bf I}\end{array}\right] {(\bar{\bf A}\bar{\bf A}^T + {\bf I}_{M-1})}^{-1} \\ \left[\begin{array}{c|c}\bar{\bf A}{\bf Q} & -{\bf I}\end{array}\right]  \left[\begin{array}{c}\bar{\bf C} \\ \tilde{\bf D}\end{array}\right]. 
\end{aligned}
\end{equation}
We can now re-write the objective in (\ref{eq:optCTLSprob}) and further simplify it to compute the minimum-norm square solution as 
\begin{equation*}
 \label{eq:TLSproof2}
 \begin{aligned}
&\mathrm{tr}\left(\left[\begin{array}{c|c}\bar{\bf E}^T & \tilde{\bf N}^T\end{array}\right] \left[\begin{array}{c}\bar{\bf E} \\ \tilde{\bf N}\end{array}\right]\right)
 = \mathrm{tr} (\left[\begin{array}{c|c}\bar{\bf C}^T & \tilde{\bf D}^T\end{array}\right]  \left[\begin{array}{c}{\bf Q}\bar{\bf A}^T \\ -{\bf I}\end{array}\right]  \\
 &\hspace*{20mm} {(\bar{\bf A}\bar{\bf A}^T + {\bf I}_{M-1})}^{-1}\left[\begin{array}{c|c}\bar{\bf A}{\bf Q} & -{\bf I}\end{array}\right]  \left[\begin{array}{c}\bar{\bf C} \\ \tilde{\bf D}\end{array}\right])\\
&\hspace*{25mm}= {\|{\boldsymbol \Lambda}^{-1/2}(\bar{\bf A}{\bf Q}\bar{\bf C} - \tilde{\bf D})\|_F^2}
\end{aligned}
\end{equation*}  
where  ${\boldsymbol \Lambda} = (\bar{\bf A}\bar{\bf A}^T + {\bf I}_{M-1})$.
Hence, the solution to the optimization problem \eqref{eq:optCTLSprob} is equivalent to the weighted OUC-LS of \eqref{eq:optCTLSprob1}.
\end{IEEEproof}

The optimization problem \eqref{eq:optCTLSprob1} does not have a closed-form solution, and has to be solved iteratively using for instance Newton's method (summarized in Algorithm~\ref{alg:newton}) with
\begin{equation}
\begin{aligned}
f({\bf Q}) &:= (\bar{\bf C}^T \otimes {\boldsymbol \Lambda}^{-1/2}\bar{\bf A}) \mathrm{vec}({\bf Q}) \in \mathbb{R}^{K \times 1},\\
\text{and} \qquad{\bf b} &:= \mathrm{vec}({\boldsymbol \Lambda}^{-1/2} \tilde{\bf D}) \in \mathbb{R}^{K \times 1}.
\end{aligned}
\end{equation}

\section{Unitarily constrained Cram\'er-Rao bound}\label{sec:CRB}
Suppose we want to estimate the vector ${\bf q}_e = [{\bf q}_1^T, {\bf q}_2^T, {\bf q}_3^T, {\bf t}^T]^T \in \mathbb{R}^{12 \times 1}$ from the measurement vector
\begin{equation}
\label{eq:crbvec}
\bar{\bf d}= ({\bf C}_e^T \otimes \bar{\bf A}) {{\bf q}_e}  + \bar{\bf n} 
\end{equation}
corrupted by noise $\bar{\bf n}$. Assume that the probability density function (PDF) $p(\bar{\bf d}; {\bf q}_e)$ of the sample vectors parameterized by the unknown vector ${\bf q}_e$ is known. The covariance matrix of any unbiased estimate of the parameter vector ${\bf q}_e$ then satisfies~\cite{SKayestimation}
\begin{equation}
\begin{aligned}
\label{eq:crbtheorem}
\mathbb{E}\{(\hat{\bf q}_e- {\bf q}_e)(\hat{\bf q}_e-{\bf q}_e)^T\} \geq {\bf C}_{CRB}({\bf q}_e) = {\bf F}^{-1}
\end{aligned}
\end{equation}
where the entries of the Fisher information matrix (FIM) ${\bf F}$ are given by
\begin{equation*}
\begin{aligned}
\label{eq:fim}
{F}_{ij} = - \mathbb{E} \left\{\frac{\partial^2 \ln p(\bar{\bf d}; {\bf q}_e)}{\partial {q}_{ei}\partial{q}_{ej}}\right\}.
\end{aligned}
\end{equation*}
This is the Cram\'er-Rao bound theorem and ${\bf C}_{CRB}$ is the Cram\'er-Rao lower bound (CRB). 

The computation of the CRB is straightforward when the noise $\bar{\bf n}$, and hence the PDF $p(\bar{\bf d}; {\bf q}_e)$ can be described by a Gaussian process. Since the noise vector $\bar{\bf n}$ is zero-mean with covariance matrix  equal to an identity matrix, the FIM can be computed using the Jacobian matrix ${\bf J}$, and is given by ${\bf F} = {\bf J}^T{\bf J} \in \mathbb{R}^{12 \times 12}$, where the Jacobian matrix is
\begin{equation*}
\begin{aligned}
\label{eq:jacobian}
{\bf J} &= \frac{\partial(\bar{\bf d} - ({\bf C}_e^T \otimes \bar{\bf A}) {{\bf q}_e})}{\partial {\bf q}_e^T} = [{\bf J}_{\bf Q} \; | \; {\bf J}_{\bf t}]  \in \mathbb{R}^{(M-1)N \times 12},\\
\end{aligned}
\end{equation*} with ${\bf J}_{\bf Q} = {\bf C}^T \otimes \bar{\bf A}$ and ${\bf J}_{\bf t} = \bar{\bf A}$. 
The FIM can then be computed as follows
\begin{equation}
\begin{aligned}
\label{eq:unfim}
{\bf F} = \left[\begin{array}{c|c}{\bf C}{\bf C}^T \otimes \bar{\bf A}^T\bar{\bf A} & ({\bf C} \otimes \bar{\bf A}^T)\bar{\bf A}\\ \hline \strut \strut \bar{\bf A}^T({\bf C}^T \otimes \bar{\bf A}) & \bar{\bf A}^T\bar{\bf A}\end{array}\right].
\end{aligned}
\end{equation}

However, note that in \eqref{eq:unfim}, the FIM does not take into account the unitary constraint on the matrix ${\bf Q}$, i.e., ${\bf Q}^T{\bf Q} = {\bf I}$. Generally, if the parameter vector ${\bf q}_e$ is subject to $K$ continuously differentiable constraints ${\bf g}({\bf q}_e) = {\bf 0}$, then with these constraints, the resulting constrained CRB is lower than the unconstrained CRB. In~\cite{constrainedCRB1}, it is shown that the constrained CRB (C-CRB) has the form
\begin{equation}
\begin{aligned}
\label{eq:constraincrb}
{\bf C}_{C-CRB}({\bf q}_e)=\mathbb{E}\{(\hat{\bf q}_e- {\bf q}_e)(\hat{\bf q}_e-{\bf q}_e)^T\}  \geq {\bf U}({\bf U}^T{\bf F} {\bf U})^{-1}{\bf U},
\end{aligned}
\end{equation}
where ${\bf F}$ is the FIM for the unconstrained estimation problem as in \eqref{eq:unfim}, and the unitary matrix ${\bf U} \in \mathbb{R}^{12 \times (12-K)}$ is obtained by collecting  orthonormal basis vectors of the null-space of the gradient matrix
\begin{equation}
\begin{aligned}
\label{eq:gradientmatrix}
{\bf G} ({\bf q}_e) = \frac{\partial \bar{\bf g}({\bf q}_e)}{\partial {\bf q}_e^T} \in \mathbb{R}^{K \times 12},
\end{aligned}
\end{equation}
where the constraints $\bar{\bf g}({\bf q}_e) = {\bf 0}$ are obtained by discarding the redundant constraints (if any) from ${\bf g}({\bf q}_e) = {\bf 0}$. This ensures that the matrix ${\bf G} ({\bf q}_e)$ is full row-rank, and implies ${\bf G} ({\bf q}_e) {\bf U} = {\bf 0}$ while ${\bf U}^T{\bf U} = {\bf I}$. For the unitarily constrained CRB (UC-CRB) denoted by ${\bf C}_{UC-CRB} ({\bf q}_e)$, we have to consider the unitary constraint ${\bf Q}^T{\bf Q} = {\bf I}$, which can be written by the following parametric constraints as
\begin{equation}
\begin{aligned}
\label{eq:constraints}
{\bf g}({\bf q}_e) = &[{\bf q}_1^T{\bf q}_1-1, {\bf q}_2^T{\bf q}_1, {\bf q}_3^T{\bf q}_1,{\bf q}_1^T{\bf q}_2, {\bf q}_2^T{\bf q}_2-1, \\
&{\bf q}_3^T{\bf q}_2,{\bf q}_1^T{\bf q}_3,{\bf q}_2^T{\bf q}_3,{\bf q}_3^T{\bf q}_3-1]^T = {\bf 0} \in \mathbb{R}^{9 \times 1}.
\end{aligned}
\end{equation}
The orthogonality constraints are symmetric, i.e., ${\bf q}_i^T{\bf q}_j ={\bf q}_j^T{\bf q}_i, i,j=1,2,3$, and hence, they are redundant. The non-redundant constraints are thus given by
\begin{equation}
\begin{aligned}
\label{eq:nonredundant}
\bar{\bf g}({\bf q}_e) = &[{\bf q}_1^T{\bf q}_1-1, {\bf q}_2^T{\bf q}_1, {\bf q}_3^T{\bf q}_1, {\bf q}_2^T{\bf q}_2-1,\\
 &{\bf q}_3^T{\bf q}_2,{\bf q}_3^T{\bf q}_3-1]^T = {\bf 0} \in \mathbb{R}^{6 \times 1}.
\end{aligned}
\end{equation}
The gradient matrix for the $K=6$ non-redundant constraints in \eqref{eq:nonredundant} can be computed as follows  
\begin{equation}
\begin{aligned}
\label{eq:gradientmatrix}
{\bf G} ({\bf q}_e) &= \frac{\partial \bar{\bf g}({\bf q}_e)}{\partial {\bf q}_e^T}  \\
&=\left[\begin{array}{cccc}2{\bf q}_1^T & {\bf 0}_3^T & {\bf 0}_3^T & {\bf 0}_3^T \\{\bf q}_2^T & {\bf q}_1^T & {\bf 0}_3^T & {\bf 0}_3^T \\{\bf q}_3^T & {\bf 0}_3^T & {\bf q}_1^T & {\bf 0}_3^T \\{\bf 0}_3^T & 2{\bf q}_2^T & {\bf 0}_3^T & {\bf 0}_3^T \\{\bf 0}_3^T & {\bf q}_3^T & {\bf q}_2^T & {\bf 0}_3^T \\{\bf 0}_3^T & {\bf 0}_3^T & 2{\bf q}_3^T & {\bf 0}_3^T\end{array}\right] \in \mathbb{R}^{6 \times 12}.
\end{aligned}
\end{equation}
An orthonormal basis of the null-space of the gradient matrix is finally given by
\begin{equation}
\begin{aligned}
\label{eq:gradientmatrix}
{\bf U} &= \frac{1}{\sqrt{2}}\left[\begin{array}{ccc|c}-{\bf q}_3 & {\bf 0}_3 & {\bf q}_2 \\{\bf 0}_3 & -{\bf q}_3 & -{\bf q}_1&{\bf 0}_{3 \times 3} \\{\bf q}_1 & {\bf q}_2 & {\bf 0}_3 \\ \hline & {\bf 0}_{3 \times 3}& & \sqrt{2} \,{\bf I}_3\end{array}\right].
\end{aligned}
\end{equation}

\begin{mylem}[Biased estimator]\label{lem:bias}
An unbiased constrained estimator for {\bf Q} does not exist, except for the noiseless case.
\end{mylem}
\begin{IEEEproof}
We prove the above claim by contradiction. Let there exist an unbiased constrained estimator $\hat{\bf Q}$ such that $\hat{\bf Q} \in \mathcal{V}_{3,3}$. Then $\hat{\bf Q} = {\bf Q} + {\boldsymbol \xi}$ where ${\boldsymbol \xi}$ is the estimation error such that $\mathbb{E}\{\hat{\bf Q}\} = {\bf Q}$ or $\mathbb{E}\{{\boldsymbol \xi}\} = 0$. Since, $\hat{\bf Q} \in \mathcal{V}_{3,3}$ we have $\hat{\bf Q}\hat{\bf Q}^T = {\bf I}_3
$, and hence
\begin{equation}
\label{eq:bias}
({\bf Q} + {\boldsymbol \xi})({\bf Q} + {\boldsymbol \xi})^T = {\bf I}_3.
\end{equation} 
Using ${\bf Q}{\bf Q}^T = {\bf I}_3$ and taking expectations on both sides, \eqref{eq:bias} can be further simplified to
\begin{equation}
\label{eq:bias2}
{\mathrm{tr}(\mathbb{E}\{\boldsymbol \xi\} {\bf Q}^T)} + \mathrm{tr}({\bf Q}\mathbb{E}\{{\boldsymbol \xi}^T\}) = - \mathrm{tr}(\mathbb{E}\{{\|{\boldsymbol \xi}\|}^2\}).
\end{equation} 
Due to the assumption that $\mathbb{E}({\boldsymbol \xi}) = 0$, the right-hand side of \eqref{eq:bias2} is zero, but, the left-hand side is strictly less than zero. Hence a contradiction occurs, unless the noise is zero.
\end{IEEEproof}

However, under Gaussian noise assumptions, and due to the asymptotic properties of a maximum likelihood (ML) estimator~\cite{SKayestimation}, at large reference ranges (low noise variances), the bias tends to zero, and the OUC-LS meets the UC-CRB. A similar argument can be found in~\cite{Whiteningrotations}, but in the context of blind channel estimation.

The UC-CRB for TA-localization can be derived from the matrix ${\bf C}_{UC-CRB}$ using the transformation of parameters. The absolute position of the sensors is a linear function of the unknown parameter vector ${\bf q}_e$, and is given by ${\bf s} = \mathrm{vec}({\bf S}) = ({\bf C}_e^T \otimes {\bf I}_3){\bf q}_e$.
The proposed TA-localization estimate is given by
\begin{equation}
\hat{\bf s}_{TA} = \mathrm{vec}(\hat{\bf S}_{TA}) = ({\bf C}_e^T \otimes {\bf I}_3)\hat{\bf q}_e.
\end{equation}
Then the UC-CRB is given by~\cite{SKayestimation}
\begin{equation}
\begin{aligned}
{\bf C}_{UC-CRB}({\bf s}) &= \frac{\partial {\bf s}}{\partial {\bf q}_e} {\bf C}_{UC-CRB}({\bf q}_e)\frac{\partial {\bf s}^T}{\partial {\bf q}_e}\\
&= ({\bf C}_e^T \otimes {\bf I}_3) {\bf C}_{UC-CRB}({\bf q}_e) ({\bf C}_e \otimes {\bf I}_3).
\end{aligned}
\end{equation}

\section{Simulation results}\label{sec:sim}

We consider $N=10$ sensors mounted along the edges of a rigid body (rectangle based pyramid of size $5 (l) \times 5 (w)\times 5 (h)~\mathrm{m}$ as in Fig.~\ref{fig:model}), and $M=4$ anchors deployed uniformly at random within a range of $1~\mathrm{km}$. The rotation matrix ${\bf Q}$ is generated with rotations of $20~\mathrm{deg}$, $-25~\mathrm{deg}$, and $10~\mathrm{deg}$ in each dimension. We use a translation vector ${\bf t} = [100,100,55]~\mathrm{m}$. The simulations are averaged over $N_{exp} = 2000$ independent Monte-Carlo experiments.

The performance of the proposed estimators is analyzed in terms of the root-mean-square-error (RMSE) of the estimates $\hat{\bf Q}$ and $\hat{\bf t}$, and are respectively given as
\begin{eqnarray}
\mathrm{RMSE}({\bf Q}) 
&=& \sqrt{\frac{1}{N_{exp}} \sum_{n=1}^{N_{exp}} {\| {\bf Q} - \hat{\bf Q}^{(n)}\|}_F^2}\nonumber\\
\text{and} \quad \mathrm{RMSE}({\bf t}) 
&=& \sqrt{\frac{1}{N_{exp}} \sum_{n=1}^{N_{exp}} {\| {\bf t} - \hat{\bf t}^{(n)}\|}_2^2},\nonumber
\end{eqnarray}
where $\hat{\bf Q}^{(n)}$ and $\hat{\bf t}^{(n)}$ denote the estimates during the $n$th Monte-Carlo experiment. To analyze the performance of the orientation estimates we introduce one more metric called the mean-angular-error (MAE) which is computed using the trace inner product, and is given by
\begin{equation}
\label{eq:MAE}
\begin{aligned}
\mathrm{MAE}({\bf Q}) 
&=
\begin{cases}
\sqrt{\frac{1}{N_{exp}} \sum_{n=1}^{N_{exp}}\mathrm{tr}(\arccos({\bf Q}^T\hat{\bf Q}^{(n)})}, & \text{if } \hat{\bf Q} \in \mathcal{V}_{3,3}\\
\sqrt{\frac{1}{N_{exp}} \sum_{n=1}^{N_{exp}}\mathrm{tr}(\arccos({\bf Q}^T\hat{\bf Q}_{norm}^{(n)})}, & \text{if } \hat{\bf Q} \notin \mathcal{V}_{3,3}
\end{cases},
\end{aligned}
\end{equation}
where we normalize the columns of $\hat{\bf Q}^{(n)}$ as $\hat{\bf Q}_{norm} = [\frac{\hat{\bf q}_1}{{\|\hat{\bf q}_1\|}_2},\frac{\hat{\bf q}_2}{{\|\hat{\bf q}_2\|}_2},\frac{\hat{\bf q}_3}{{\|\hat{\bf q}_3\|}_2}]$ when $\hat{\bf Q} \notin \mathcal{V}_{3,3}$, and as earlier $\hat{\bf Q}^{(n)}$ and $\hat{\bf Q}_{norm}^{(n)}$ correspond to the estimate obtained during the $n$th Monte-Carlo experiment. The normalization is done for the estimates based on the unconstrained LS, as in this case the estimated ${\bf Q}$ matrix is not necessarily orthogonal.

Simulations are provided for different values of the {\it reference range} $\zeta$. In the considered example, the maximum range is around $700~\mathrm{m}$, hence, a reference range of $80~\mathrm{dB}$ corresponds to $\frac{700}{\sqrt{10^8}} = 0.07~\mathrm{m}$ error (standard deviation) on the range measurements.

In Fig.~\ref{fig:LSMSEQ}, the RMSE of the estimated ${\bf Q}$ matrix is illustrated for the proposed estimators when the topology of the sensors is accurately known. The unconstrained LS estimator is efficient, and meets the (unconstrained) root CRB (RCRB). However, the solution of the unconstrained LS $\hat{\bf Q}_{LS}$ need not be necessarily an orthogonal matrix. The performance of the SUC-LS estimator is similar (slightly worse) to that of the iterative OUC-LS. However, OUC-LS is efficient and meets the CRB at reasonable values of the reference range. The bias of both the SUC-LS and OUC-LS estimators is shown in Fig.~\ref{fig:bias}, and it can be seen that the bias tends to zero for $\zeta > 50~\mathrm{dB}$ (as discussed in Lemma~\ref{lem:bias}), whereas the unconstrained LS is an unbiased estimator. The bias is computed as follows
\begin{equation*}
\mathrm{Bias}({\bf Q}) = \|{\frac{1}{N_{exp}} \sum_{n=1}^{N_{exp}} \mathrm{vec}(\hat{\bf Q}^{(n)}) - \mathrm{vec}({\bf Q})}\|_2.
\end{equation*}

\begin{figure} [!t]
\centering
\includegraphics[width=3.5in]{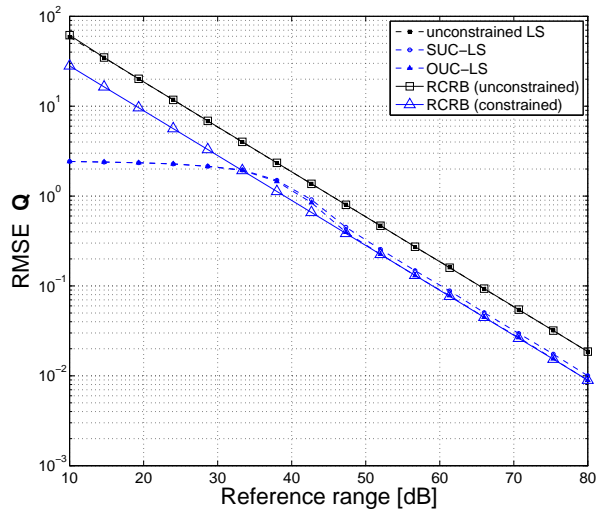}
\caption{RMSE of the estimated rotation matrix ${\bf Q}$.} 
\label{fig:LSMSEQ}
\end{figure}
\begin{figure} [!t]
\centering
\includegraphics[width=3.5in]{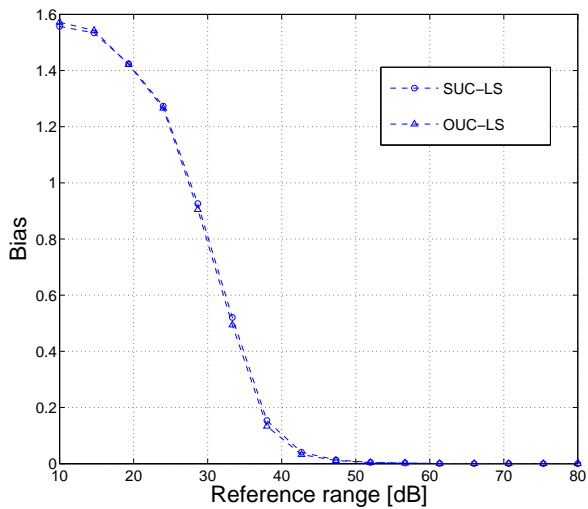}
\caption{Bias in the SUC-LS and OUC-LS estimators for ${\bf Q}$.} 
\label{fig:bias}
\end{figure}
\begin{figure} [!t]
\centering
\includegraphics[width=3.5in]{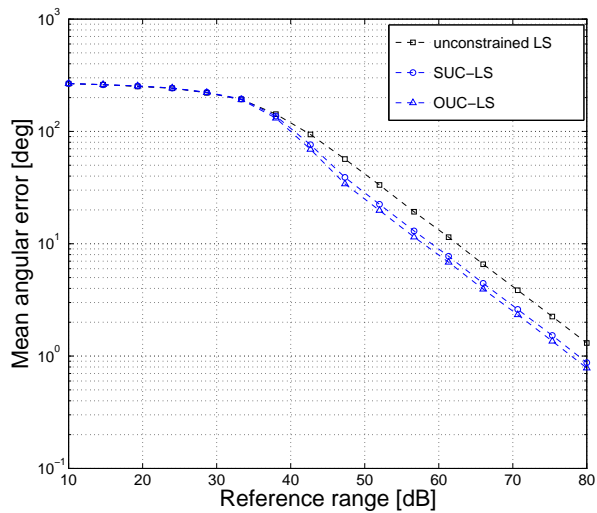}
\caption{MAE of the estimated rotation matrix ${\bf Q}$.} 
\label{fig:LSMAEQ}
\end{figure}
\begin{figure} [!t]
\centering
\includegraphics[width=3.5in]{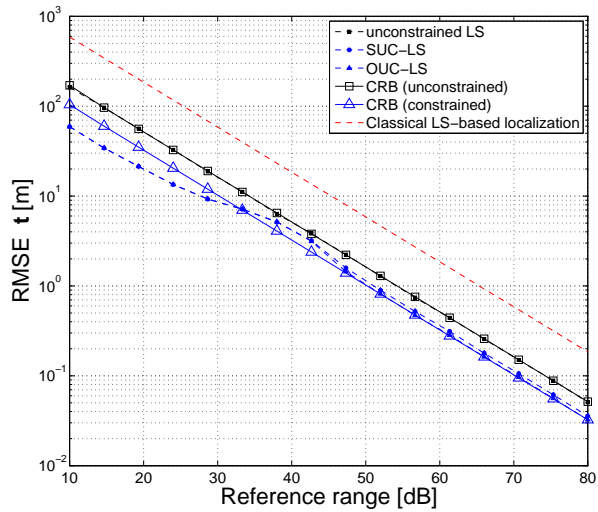}
\caption{RMSE of the estimated translation vector ${\bf t}$ along with the solution from the classical LS-based localization.} 
\label{fig:LSMSEt}
\end{figure}
\begin{figure} [!t]
\centering
\includegraphics[width=3.5in]{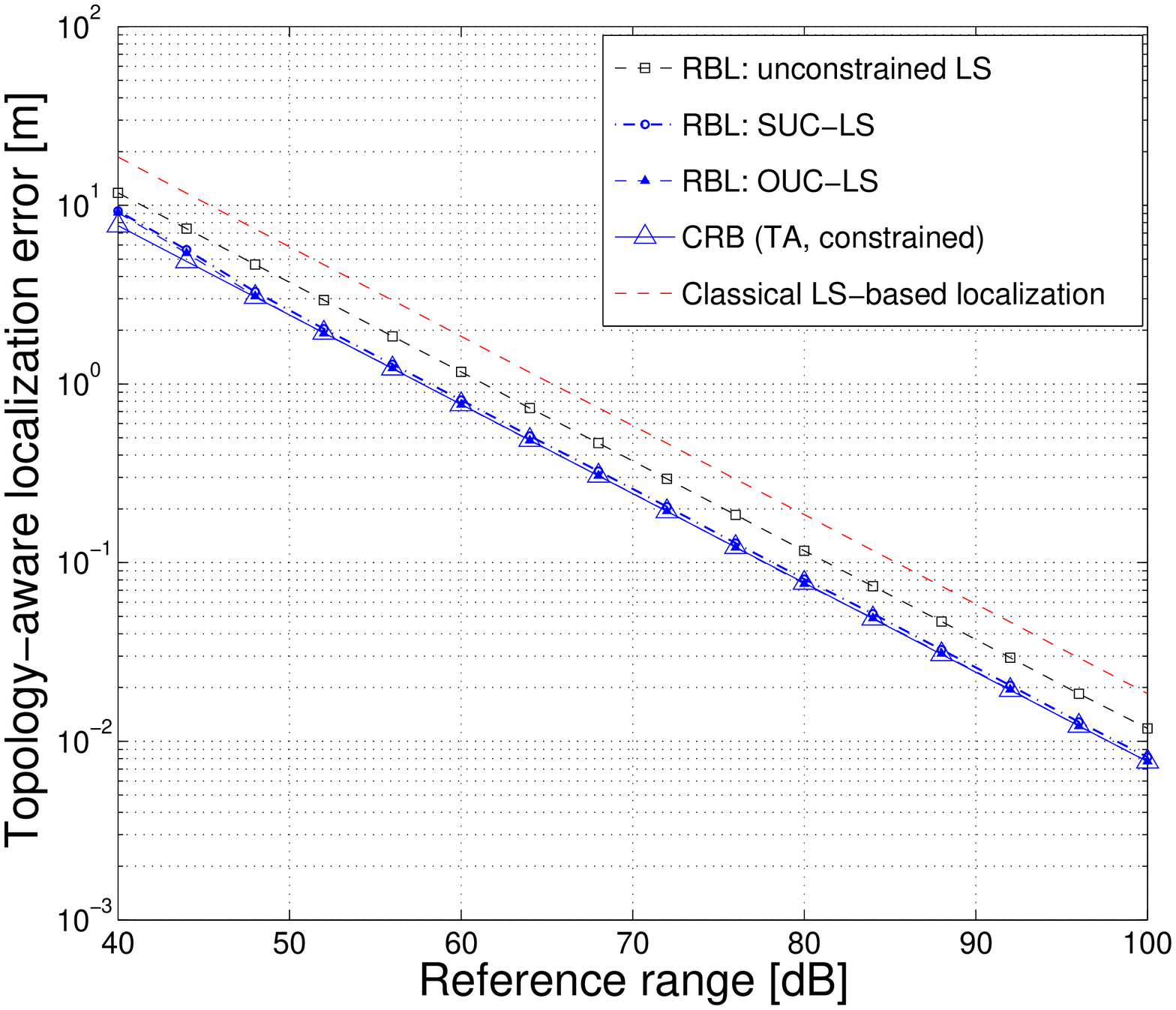}
\caption{RMSE for TA-localization.} 
\label{fig:localization}
\end{figure}

\begin{myrem}[Frobenius norm induced distance] \label{rem:Forb}
For any matrix ${\bf Q}_i$ and ${\bf Q}_j$, such that, ${\bf Q}_i \in \mathcal{V}_{n,n}$ and ${\bf Q}_j \in \mathcal{V}_{n,n}$, the Frobenius norm induced distance is always upper bounded by $\sqrt{2n}$, i.e., ${\|{\bf Q}_i - {\bf Q}_j\|}_F \leq \sqrt{{\|{\bf Q}_i \|}_F + {\| {\bf Q}_j\|}_F} = \sqrt{2n}$. 
\end{myrem}
The saturation of the RMSE in Fig.~\ref{fig:LSMSEQ} for $\zeta < 30~\mathrm{dB}$ follows from Remark~\ref{rem:Forb}, and yields a low RMSE due to the bias. However, the UC-CRB computed using \eqref{eq:constraincrb} does not saturate in this range. Fig.~\ref{fig:LSMAEQ} shows the MAE, which gives an insight in how the error on the range measurements translates to the error on the estimated rotations. For the unconstrained LS, the MAE is computed based on normalization as discussed earlier in  \eqref{eq:MAE}.

Fig.~\ref{fig:LSMSEt} shows the RMSE of the estimated translation vector for the estimators based on the accurate knowledge of ${\bf C}$. The translation vector corresponds to a single three-dimensional absolute position of the rigid body, and has a significant (close to an order of magnitude) performance improvement compared to the classical LS-based localization for the considered scenario. This is due to the error involved in estimating $N$ locations independently. The RMSE for the classical LS-based localization is computed as
\begin{equation}
\begin{aligned}
\mathrm{RMSE}({\bf S}) 
&=\sqrt{\frac{1}{N_{exp}} \sum_{n=1}^{N_{exp}} {\| {\bf S} - \hat{\bf S}_{LS}^{(n)}\|}_F^2}
\end{aligned}
\end{equation}
where $\hat{\bf S}_{LS}^{(n)}$ is the estimate during the $n$th Monte-Carlo experiment.
\begin{figure} [!t]
\centering
\includegraphics[width=3.5in]{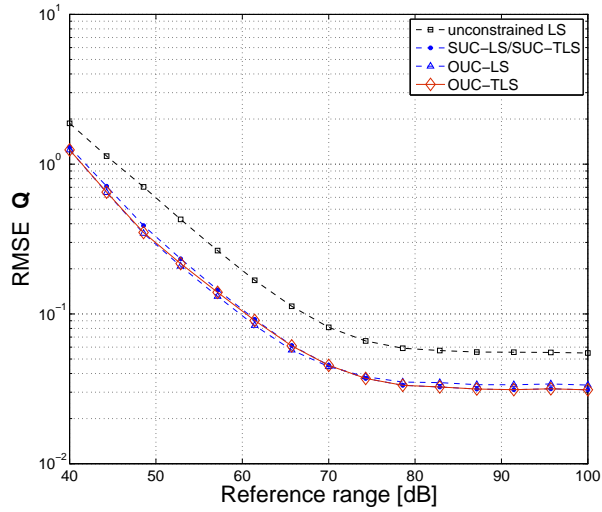}
\caption{RMSE of the estimated rotation matrix ${\bf Q}$ with perturbed ${\bf C}$.} 
\label{fig:TLSMSEQ}
\end{figure}
\begin{figure} [!t]
\centering
\includegraphics[width=3.5in]{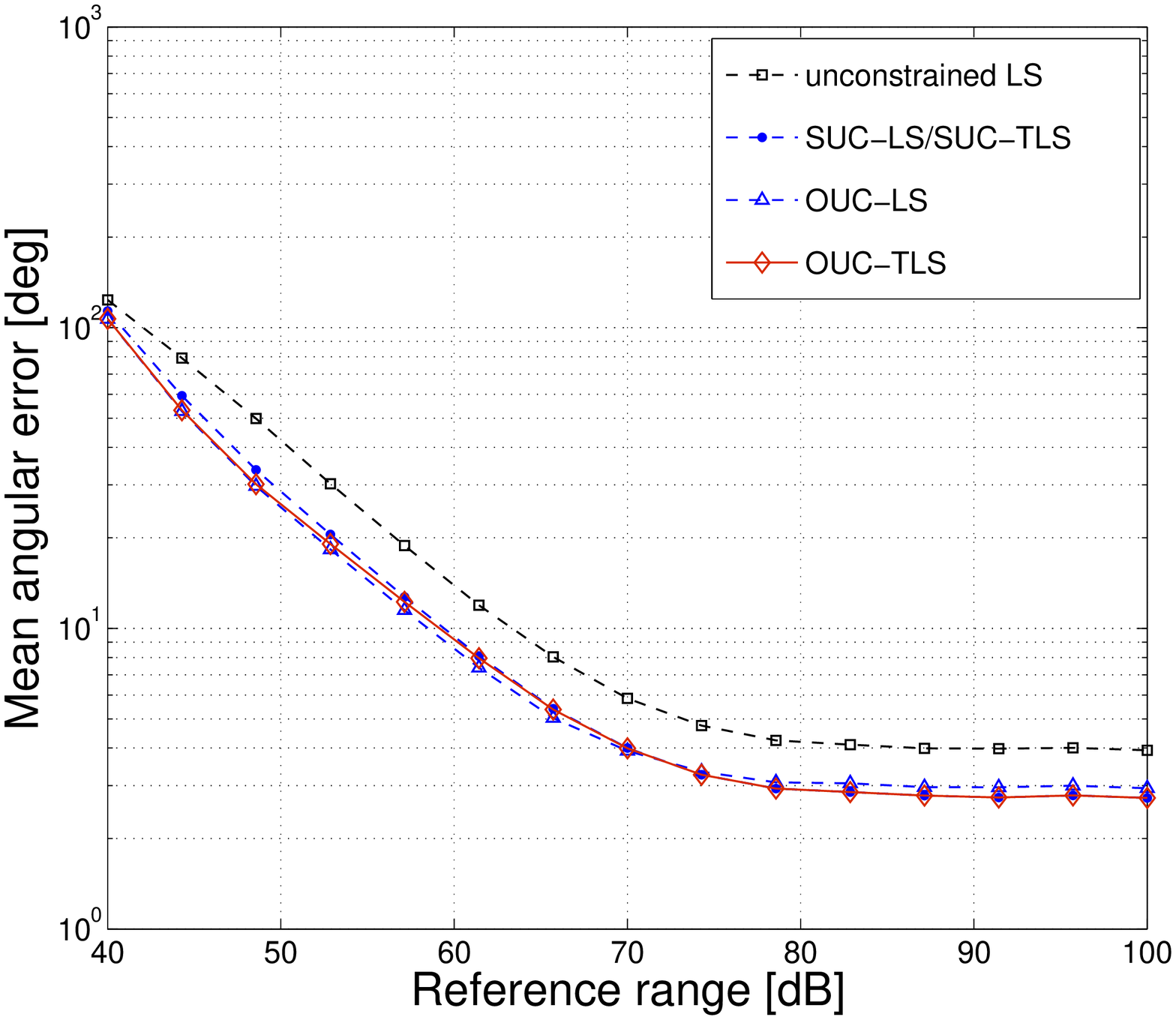}
\caption{MAE of the estimated rotation matrix ${\bf Q}$ with perturbed ${\bf C}$.} 
\label{fig:TLSMAEQ}
\end{figure}
\begin{figure} [!t]
\centering
\includegraphics[width=3.5in]{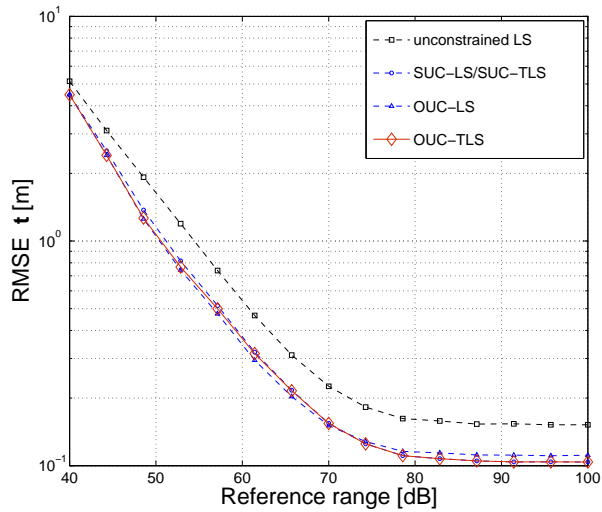}
\caption{RMSE of the estimated translation vector ${\bf t}$ with perturbed ${\bf C}$.} 
\label{fig:TLSMSEt}
\end{figure}
The locations of the sensors mounted on the rigid body can be estimated from $\hat{\bf Q}$ and $\hat{\bf t}$, which is known as TA-localization. The improvement in the localization performance of the rigid body localization algorithms as compared to the classical LS-based localization can be seen in Fig.~\ref{fig:localization}, and the improvement in the localization performance is due to the knowledge of the sensor topology.

In order to analyze the performance of the estimators for the case when the sensor topology is perturbed, we corrupt the sensor coordinates in the reference frame with a zero mean i.i.d. Gaussian random process of standard deviation $\sigma_e = 10~\mathrm{cm}$, i.e., ${\bf e}_n \sim \mathcal{N}(0, \sigma_e^2{\bf I}_3)$ for $n=1,2,\ldots,N$.

The RMSE of the estimated $\hat{\bf Q}$, $\hat{\bf t}$ using the unconstrained LS, SUC-LS/SUC-TLS, OUC-LS and OUC-TLS estimators is shown in Fig.~\ref{fig:TLSMSEQ} and Fig.~\ref{fig:TLSMSEt}, respectively. The performance of these estimators is similar to that of the LS-based estimators, except for the error floor, and this is due to the model error (perturbations on the sensor topology). The MAE for the SUC-TLS and OUC-TLS estimators is shown in Fig.~\ref{fig:TLSMAEQ}.

\section{Conclusions} \label{sec:con}
A novel framework for joint position and orientation estimation of a rigid body based on range-only measurements is proposed. We refer to this problem as rigid body localization. Sensor nodes can be mounted on the rigid bodies (e.g., satellites, robots)  during fabrication, and the geometry of how these sensors are mounted is known {\it a priori} up to a certain accuracy. However, the absolute position of the sensors or the rigid body itself is not known. Using the range measurements between the anchors and the sensors on the rigid body, as in classical localization schemes, we can estimate the position and the orientation of the body. This is equivalent to estimating a rotation matrix and a translation vector with which we parameterize the Stiefel manifold. The problem can also be viewed as localizing sensors with a manifold constraint (e.g., the sensors lie on a rigid body), and with this additional information the performance naturally improves. The constrained Cram\'er-Rao bounds are derived as a benchmark for the proposed estimators. Estimators that take into account the inaccuracies in the known sensor topology are also proposed.
\appendices
\section{Covariance matrices of error vectors} \label{app:covmatrix}
The covariance matrix of the noise $\bar{\bf n}$ can be computed as follows
\begin{equation}
\begin{aligned}
\label{eq:covnbar}
\mathbb{E}\{\bar{\bf n}\bar{\bf n}^T\} &= \mathbb{E}\{\mathrm{vec}({\bf U}_M^T {\bf W}{\bf N})\mathrm{vec}({\bf U}_M^T {\bf W}{\bf N})^T\}\\
&=\mathbb{E}[({\bf I}_N \otimes {\bf U}_M^T{\bf W}) \mathrm{vec}({\bf N})\mathrm{vec}({\bf N})^T ({\bf I}_N \otimes {\bf W}^T{\bf U}_M)] \\
&=({\bf I}_N \otimes {\bf U}_M^T{\bf W}) \mathbb{E}\{\mathrm{vec}({\bf N})\mathrm{vec}({\bf N})^T\} ({\bf I}_N \otimes {\bf W}^T{\bf U}_M)\\
&=({\bf I}_N \otimes {\bf U}_M^T{\bf W})({\bf I}_N \otimes {\boldsymbol \Sigma}) ({\bf I}_N \otimes {\bf W}^T{\bf U}_M)\\
&=({\bf I}_N \otimes {\bf U}_M^T{\bf W}{\boldsymbol \Sigma}{\bf W}^T{\bf U}_M)\\
&\approx{\bf I}_{(M-1)N},
\end{aligned}
\end{equation}
where the last approximate equality is due to the estimated pre-whitening matrix. The covariance matrix of the noise $\tilde{\bf n}$ can be computed along similar lines, and hence it is not presented here.

\section{Newton's method} \label{app:newton}
The initial point for the Newton's algorithm is computed by solving the following equality constrained LS problem~\cite{CLSQuad}
\begin{equation}
\label{eq:init1}
\begin{aligned} 
\check{\bf Q}_0 &= \argmin_{Q} {\|f({\bf Q}) -{\bf b}\|}_2^2 \\
&s.t. \quad {\|{\bf q}\|}_2 = \sqrt{3}
\end{aligned}
\end{equation} where ${\bf q} = \mathrm{vec}({\bf Q})$. Since $\check{\bf Q}_0$ does not necessarily have orthonormal columns, the OPP (similar to  \eqref{eq:CLSprob}) is solved to obtain the initial value for the Newton's method
\begin{equation}
\label{eq:init2}
\begin{aligned}
{\bf Q}_0  &= \argmin_{\bf Q} {\|{\bf Q} - \check{\bf Q}_0\|}_F^2 \quad s.t. \quad {\bf Q}^T{\bf Q} = {\bf I}_3 \\
&=( \check{\bf Q}_0 \check{\bf Q}_0^T)^{-1/2}  \check{\bf Q}_0.
\end{aligned}
\end{equation}

For an unconstrained minimization problem, the Newton's method is generally derived using a second-order Taylor series expansion of the cost function around a point. For optimizations involving unitary constraints, we can parameterize the unitary matrix ${\bf Q}$ using a matrix exponential function of a skew-symmetric matrix (sometimes also referred to as the matrix Lie algebra of $\mathcal{V}_{3,3}$~\cite{rigidbodyControlmag})
\begin{equation}
{\bf X}({\bf x}) = \left[\begin{array}{ccc}0 & -x_1 & -x_2 \\x_1 & 0 & -x_3 \\x_2 & x_3 & 0\end{array}\right] \in \mathbb{R}^{3 \times 3},
\end{equation} 
where ${\bf X} = - {\bf X}^T$, and ${\bf x}=[x_1,x_2,x_3]^T$. 

Given a point $f(\breve{\bf Q})$ on the manifold $\breve{\bf Q}$, we can represent any unitary matrix ${\bf Q}$ in the vicinity of a given unitary matrix $\breve{\bf Q}$ as
\begin{equation}
{\bf Q} = \breve{\bf Q} \, \exp({\bf X}({\bf x})).
\end{equation}

To compute the Newton or a Gauss-Newton step (a descent direction) to \eqref{eq:iterMLprob}, we then use the series expansion of the matrix exponential 
\begin{equation}
{\bf Q} = \breve{\bf Q} ({\bf I} + {\bf X} + \frac{{\bf X}^2}{2!}+ \cdots),
\end{equation}
and obtain the expansion for 
\begin{equation}
\label{eq:firstorder}
\begin{aligned}
f({\bf Q}) &= f(\breve{\bf Q}) + f(\breve{\bf Q}{\bf X}) + f(\breve{\bf Q}\frac{{\bf X}^2}{2!})+ \cdots\\
&= f(\breve{\bf Q}) + {\bf J}_Q {\bf x} + \cdots
\end{aligned}
\end{equation} around $\breve{\bf Q}$, where ${\bf J}_{Q} \in \mathbb{R}^{K \times 3}$ is the Jacobian matrix. The Jacobian matrix can be expressed as column vectors corresponding to entries of $\bf x$, i.e., $x_1,x_2,x_3$ as
\begin{equation}
\begin{aligned}
{\bf J}_{Q} = \left[\begin{array}{ccc}{\bf j}_{Q,21} & {\bf j}_{Q,31} & {\bf j}_{Q,32}\end{array}\right]
\end{aligned}
\end{equation} where the column vectors are given by ${\bf j}_{Q,ij} = f(\breve{\bf Q}({\boldsymbol \delta}_i{\boldsymbol \delta}_j^T - {\boldsymbol \delta}_j{\boldsymbol \delta}_i^T))$ for appropriate values of $i$ and $j$, and the standard unit vectors ${\boldsymbol \delta}_1, {\boldsymbol \delta}_2, {\boldsymbol \delta}_3 \in \mathbb{R}^3$. 

Using the first-order approximation \eqref{eq:firstorder} in \eqref{eq:iterMLprob}, we can compute the Gauss-Newton step for solving the optimization problem, which is given by
\begin{equation}
\label{eq:GaussNewton}
\begin{aligned}
\Delta{\bf x}_{GN} &= \min_{\bf x} \, {\| f(\breve{\bf Q}) + {\bf J}_Q{\bf x} - {\bf b}\|}_2^2 \\
&= -  {\bf J}_Q^\dag (f(\breve{\bf Q}) - {\bf b}).
\end{aligned}
\end{equation}

In order to compute the full Newton search direction, the Hessian matrix (containing the second-order derivatives) ${\bf H}_Q \in \mathbb{R}^{3 \times 3}$ is needed. The Newton search direction is given by
\begin{equation}
\label{eq:Newton}
\begin{aligned}
\Delta{\bf x}_{N} &= -  ({\bf J}_Q^T{\bf J}_Q + {\bf H}_Q)^{-1}  {\bf J}_Q^T(f(\breve{\bf Q}) - {\bf b}).
\end{aligned}
\end{equation}

To compute the Hessian matrix using the term $f(\breve{\bf Q}\frac{{\bf X}^2}{2!})$, we express 
\begin{equation*}
{\bf X}^2 = \left[\begin{array}{ccc}-(x_1^2+x_2^2) & -x_2x_3 & x_1x_3 \\-x_2x_3 & -(x_1^2+x_3^2) & -x_1x_2 \\x_1x_3 & -x_1x_2 & -(x_2^2+x_3^3)\end{array}\right]
\end{equation*}
as a sum of the following six matrices
\begin{equation}
\begin{aligned}
{\bf X}^2 = x_1^2 {\bf T}_{1,1} + x_1x_2 {\bf T}_{1,2} + x_1x_3 {\bf T}_{1,3} \\ + x_2^2 {\bf T}_{2,2} 
+ x_2x_3 {\bf T}_{2,3} + x_3^2{\bf T}_{3,3} 
\end{aligned}
\end{equation}
where we have introduced the matrices
\begin{eqnarray*}
{\bf T}_{1,1} &=& -({\boldsymbol \delta}_1{\boldsymbol \delta}_1^T + {\boldsymbol \delta}_2{\boldsymbol \delta}_2^T),\\
{\bf T}_{1,2} &=& -({\boldsymbol \delta}_3{\boldsymbol \delta}_2^T + {\boldsymbol \delta}_2{\boldsymbol \delta}_3^T),\\
{\bf T}_{1,3} &=& \,\,\, \,\,({\boldsymbol \delta}_3{\boldsymbol \delta}_1^T + {\boldsymbol \delta}_1{\boldsymbol \delta}_3^T),\\
{\bf T}_{2,2} &=& - ({\boldsymbol \delta}_2{\boldsymbol \delta}_1^T + {\boldsymbol \delta}_1{\boldsymbol \delta}_2^T),\\
{\bf T}_{2,3} &=& - ({\boldsymbol \delta}_1{\boldsymbol \delta}_1^T + {\boldsymbol \delta}_3{\boldsymbol \delta}_3^T),\\
\text{and} \quad {\bf T}_{3,3} &=& -({\boldsymbol \delta}_2{\boldsymbol \delta}_2^T + {\boldsymbol \delta}_3{\boldsymbol \delta}_3^T).
\end{eqnarray*}

Now, we can express the Hessian matrix as
\begin{equation}
{\bf H}_Q = \frac{1}{2}\left[\begin{array}{ccc} 2{\bf w}^T{\bf h}_{Q,11} & {\bf w}^T{\bf h}_{Q,21} &  {\bf w}^T{\bf h}_{Q,31} \\ {\bf w}^T{\bf h}_{Q,21} &  2{\bf w}^T{\bf h}_{Q,22} & {\bf w}^T{\bf h}_{Q,32} \\ 
{\bf w}^T{\bf h}_{Q,31} &  {\bf w}^T{\bf h}_{Q,32} &  2{\bf w}^T{\bf h}_{Q,33}\end{array}\right]
\end{equation}
where the residual ${\bf w}= f(\breve{\bf Q}) - {\bf b}$, and ${\bf h}_{Q,ij} = f(\tilde{\bf Q}{\bf T}_{i,j})$ for appropriate values of $i$ and $j$.

Once the descent direction is computed based on  Gauss-Newton's step \eqref{eq:GaussNewton} or Newton's step \eqref{eq:Newton}, the {\it step-length} to move along the surface of $f({\bf Q})$ starting from $f(\breve{\bf Q})$ in the search direction is computed by solving
\begin{equation}
\label{eq:optimalsteplength}
\hat{\gamma} = \min_{\gamma \in (0,1]} \quad {\| f({\bf Q}(\gamma{\bf X}({\bf x}))) - {\bf b}\|}_2^2
\end{equation}
where ${\bf Q}(\gamma{\bf X}({\bf x})) = \breve{\bf Q} \exp(\gamma{\bf X}({\bf x}))$.

\bibliographystyle{IEEEbib}
\bibliography{IEEEabrv,strings,refs}

\end{document}